\documentclass[a4paper,fleqn]{cas-sc}

\usepackage[numbers]{natbib}

\usepackage{amsmath,amssymb,amsfonts,mathtools}
\usepackage{amsthm}
\usepackage{bm}

\usepackage{graphicx}
\usepackage{subcaption}
\usepackage{float}

\usepackage{booktabs}
\usepackage{tabularx}
\usepackage{array}
\usepackage{multirow}
\usepackage{adjustbox}

\usepackage{xspace}
\usepackage{enumitem}
\usepackage{pifont}
\usepackage[normalem]{ulem}
\usepackage{soul}
\usepackage{stackengine}

\usepackage{tikz}
\usetikzlibrary{arrows.meta,positioning,fit,calc,backgrounds}

\newtheorem{theorem}{Theorem}[section]
\newtheorem{proposition}[theorem]{Proposition}
\newtheorem{lemma}[theorem]{Lemma}
\newtheorem{corollary}[theorem]{Corollary}

\theoremstyle{definition}
\newtheorem{definition}[theorem]{Definition}
\newtheorem{example}[theorem]{Example}

\theoremstyle{remark}
\newtheorem{remark}[theorem]{Remark}

\newcommand{\NN}{\ensuremath{\mathbb{N}}\xspace}
\newcommand{\ZZ}{\ensuremath{\mathbb{Z}}\xspace}

\newcommand{\RR}{\ensuremath{\mathbb{R}}\xspace}

\DeclareMathOperator{\Dom}{Dom}

\newcommand{\TP}{\ensuremath{T_{\bm{P}}}\xspace}

\newcommand{\TM}{\ensuremath{T_{\bm{M}}}\xspace}
\newcommand{\TLK}{\ensuremath{T_{\bm{LK}}}\xspace}

\newcommand{\TA}{\ensuremath{T_A}\xspace}

\newcommand{\TB}{\ensuremath{T_B}\xspace}

\definecolor{asparagus}{rgb}{0.53, 0.66, 0.42}

\definecolor{darkamber}{HTML}{8C4A16}

\begin{document}

\let\WriteBookmarks\relax

\def\floatpagepagefraction{1}
\def\textpagefraction{.001}

\shorttitle{Reconstruction of continuous Archimedean t-norms}

\shortauthors{Bauza-Llompart et al.}

\title[mode=title]{
Reconstruction of continuous Archimedean t-norms from their section along an increasing curve
}


\author[1,3,4]{Gabriel Bauza-Llompart}
\ead{[gabriel.bauza@uib.cat](mailto:gabriel.bauza@uib.cat)}
\cormark[1]

\author[2]{Raquel Fernandez-Peralta}
\ead{[raquel.fernandez@mat.savba.sk](mailto:raquel.fernandez@mat.savba.sk)}

\author[1,3,4]{Marc Munar}
\ead{[marc.munar@uib.es](mailto:marc.munar@uib.es)}


\affiliation[1]{
organization={
Soft Computing, Image Processing and Aggregation (SCOPIA)
Research Group, Department of Mathematics and Computer Science,
University of the Balearic Islands
},
addressline={Ctra. de Valldemossa, km 7.5},
city={Palma},
postcode={07122},
country={Spain}
}

\affiliation[2]{
organization={
Mathematical Institute, Slovak Academy of Sciences
},
city={Bratislava},
country={Slovakia}
}

\affiliation[3]{
organization={
Balearic Islands Health Research Institute (IdISBa)
},
city={Palma},
postcode={07010},
country={Spain}
}

\affiliation[4]{
organization={
Artificial Intelligence Research Institute
of the Balearic Islands (IAIB)
},
city={Palma},
postcode={07122},
country={Spain}
}

\cortext[1]{Corresponding author.}

\begin{abstract}
The completion problem for continuous and Archimedean triangular norms (t-norms
for short) consists in describing all the operators that take predefined values
on a given subset of the unit square, and it reveals how much of a t-norm is
determined by partial information. In this paper, we focus on the case where this
subset is a parametric curve. We first formulate the problem in terms of the
functional equation satisfied by an additive generator, and we then restrict our
attention to curves of the form $(x,p(x))$ with $p$ continuous and increasing,
which include the classical diagonal and horizontal section cases. We determine
when the known values are compatible with a continuous and Archimedean t-norm,
and we show that every completion is obtained by choosing the additive generator
on a subinterval and reconstructing it elsewhere from the values on the curve.
This description is recursive in general and becomes a closed formula under
additional conditions, recovering the two classical results as particular cases.
The completions are not unique, and we give examples with infinitely many
t-norms sharing the same section along a curve. Finally, we combine our results
with the reconstruction of a t-norm from its values above a level curve, which
determines exactly the additive generator on that subinterval, so that the two
pieces of information are complementary and together determine a unique
completion.
\end{abstract}


\maketitle

\section{Introduction}
In many areas of mathematics, algebraic or analytic objects are not given globally, but through their behavior on restricted subsets of their domain. This scheme appears across diverse contexts: in problems where a function given only on a subset admits a global extension compatible with certain regularity or structural constraints (for instance, in classical results in topology, measure theory, or probability such as Tietze, Carathéodory or Kolmogorov extension theorems) \cite{Dudley2002}; in functional equations, where identities may be assumed only on subsets of the usual domain \cite{Kuczma1978,Sander2005}; in algebraic structures, where partially defined operations are extended to total operations satisfying certain axioms  \cite{Goralcik2006,Gowers2020,Ljapin1988}; or in inverse and reconstruction problems, where incomplete yet structured data suffice to recover complete models \cite{Isakov2017}. A unifying theme across these problems is that limited initial information, when constrained by strong structural requirements, can be propagated to yield full characterizations, sometimes even enabling explicit reconstruction.

In particular, consider the completion problem for associative operations: given a binary operation $F$ defined only on a subset $D \subseteq S \times S$ that satisfies associativity wherever defined, can $F$ be extended to all of $S \times S$ while preserving associativity? For continuous operations on real intervals, the classical theory of the associativity equation $F(F(x,y),z) = F(x,F(y,z))$
provides key tools. Aczél showed that continuous, strictly monotone solutions admit the representation $F(x,y) = f^{-1}(f(x) + f(y))$ via a generator function $f$ \cite{Aczel1966}. Such a representation reduces the inherently two-variable associativity constraint to a one-dimensional additive structure, greatly simplifying both analysis and construction. From the perspective of the completion problem, these generator representations are particularly useful: if a partially defined associative operation can be expressed via a generator on its domain, then extending the generator in a compatible way yields an associative extension of the original operation \cite{Alsina2006}.

An important setting in which these ideas appear is the theory of triangular norms (t-norms for short). A t-norm is a binary operation $T:[0,1]^2 \to [0,1]$ that is commutative, associative, monotone and has $1$ as a neutral element. They were introduced by Menger in the context of probabilistic metric spaces \cite{Menger1942} and were later axiomatized by Schweizer and Sklar \cite{Schweizer1961}. Today, t-norms play a central role in fuzzy logic, aggregation theory, and a wide range of applications, including information theory, risk analysis, and decision making \cite{eiter2025tnorm}.

The completion problem for t-norms is a specific case of the associativity equation with further conditions, which, similarly to Aczél's result, also admits a representation in terms of an additive (or multiplicative) generator when the operator is Archimedean \cite{Klement2000}. Further, continuous t-norms are completely characterized in terms of ordinal sums of Archimedean t-norms. Thanks to these characterizations, the completion problem can be reduced to the study of functional equations (which may be defined in partial domains) in the Archimedean case \cite{Alsina2006}, which can then be leveraged for the continuous case via the ordinal sum construction. The completion problem has been extensively studied for various subsets of $[0,1]^2$, producing two main types of results: (i) completion results that provide explicit reconstructions of the operator, which may or may not be unique; and (ii) uniqueness results that determine whether a completion, when it exists, is unique or not (see \cite{Fernandez-Peralta2019,Fernandez-PeraltaSub} for more information). Two canonical examples are when the t-norm is known on the diagonal $\{(x,x): x \in [0,1]\}$ or on a horizontal section $\{(x,\alpha): x \in [0,1]\}$ for some fixed $\alpha \in (0,1)$. In these cases, the completion problem reduces to solving Schröder's and Abel's functional equations, respectively, yielding explicit but non-unique solutions that depend on the arbitrary choice of a certain part of the corresponding additive generator. Beyond these classical results, the completion problem for t-norms known along arbitrary parametric curves has received little attention in the literature.

In this paper, we revisit the completion problem and study the case where a continuous and Archimedean t-norm is known only along a parametric curve $\gamma\colon I \subseteq [0,1] \to [0,1]^2$. After a general discussion of the problem for arbitrary $\gamma$, we focus on the case $\gamma(\theta)=(\theta,p(\theta))$ with $p$ a continuous increasing function, which naturally encompasses the classical diagonal and horizontal section cases as special instances. For this setting, we derive descriptions of the additive generators of all possible completions, which are recursive in general and become explicit under additional conditions on $p$, recovering in particular the classical results in both cases. As in those cases, these completions are not unique, since they depend on the choice of the restriction of the generator to an interval $[a,1]$. Nonetheless, our results can be combined with existing results for other regions to yield explicit, unique reconstructions. In particular, by integrating our approach with the reconstruction of a pre-t-norm known above a level curve \cite{Fernandez-PeraltaSub}, we prove that choosing this restriction of the generator is equivalent to knowing the operator above a level curve of value $a$, so that both pieces of information together determine a unique completion, thereby establishing a new framework that encompasses both parametric curve restrictions and set constraints.

The rest of this paper is structured as follows. Section~\ref{sec:background} presents the necessary background on t-norms, additive generators, and reviews related work on the completion problem. Section~\ref{sec:parametric} discusses the completion problem in its general form when a t-norm is partially known on an arbitrary parametric curve $\gamma$, and then focuses on the particular case $\gamma(\theta)=(\theta,p(\theta))$. Section~\ref{sec:constant} addresses the case where $p$ is bounded away from zero, in which the reconstruction is obtained directly and in closed form. Section~\ref{sec:general} extends the analysis to the case where $p$ is a suitable path, a condition that allows the recursive scheme to be carried out and that covers the previous section as a particular case. Section~\ref{sec:unique} combines our results with the reconstruction above a level curve to obtain unique completions. Finally, Section~\ref{sec:conclusions} summarizes our main contributions and discusses directions for future research.

\section{Background and related work}\label{sec:background}

In this section, we present all the relevant results and definitions to make this work self-contained.

\subsection{Triangular norms and pre-t-norms}

Triangular norms, also called t-norms, are a kind of associative, binary operation originated and widely studied in the framework of probabilistic metric spaces. Also, in fuzzy logic, a t-norm usually plays the role of the generalization of the classical conjunction.

\begin{definition}[\bf{\cite{Klement2000,Menger1942}}]\label{def:tnorm}
	A binary operator $T:[0,1]^2 \to [0,1]$ is said to be a \emph{triangular norm} (or \emph{t-norm} for short) if it satisfies:
	\begin{description}
		\item[(T1)] $T(x,y)=T(y,x)$ for all $x,y \in [0,1]$. \hfill (Commutativity)
		\item[(T2)] $T(x,T(y,z))=T(T(x,y),z)$ for all $x,y,z \in [0,1]$. \hfill (Associativity)
		\item[(T3)] $T(x,y) \leq T(x,z)$ when $y \leq z$, for all $x,y,z \in [0,1]$. \hfill (Monotonicity)
		\item[(T4)] $T(x,1)=x$ for all $x \in [0,1]$. \hfill (Boundary Condition)
	\end{description}
\end{definition}

Several structural properties are commonly used to describe and classify t-norms, as they capture different aspects of their algebraic behavior and are linked to important representation results. We therefore recall the following standard subclasses.

\begin{definition}[\bf{\cite[Definition 2.1.2]{Baczynski2008}}]\label{def:properties_tnorm}
A t-norm $T:[0,1]^2 \to [0,1]$ is called
	\begin{itemize}
		\item \emph{idempotent} if $T(x,x)=x$ for all $x \in [0,1]$;
		\item \emph{positive} if $T(x,y)=0$ implies either $x=0$ or $y=0$;
		\item \emph{cancellative} if $T(x,y)=T(x,z)$ implies $y=z$ for all $x,y,z \in [0,1]$, $x>0$;
		\item \emph{conditionally cancellative} if $T(x,y)=T(x,z)>0$ implies $y=z$ for all $x,y,z \in [0,1]$;
		\item \emph{Archimedean} if for all $x,y \in (0,1)$ there exists an $n \in \NN$ such that $x_T^{(n)}<y$, where
        $$x_T^{(n)} = \begin{cases}
            x, & \text{if $n=1$,} \\
            T\left( x, x_T^{{(n-1)}} \right), & \text{otherwise.}
        \end{cases}$$
	\end{itemize}
\end{definition}

\begin{definition}[\textbf{\cite[Definition 2.13]{Klement2000}}]  A continuous and Archimedean t-norm $T:[0,1]^2 \to [0,1]$ is called	
	\begin{itemize}
		\item \emph{strict} if it is strictly increasing on $(0,1]^2$.
		\item \emph{nilpotent} if there exists $(x,y) \in (0,1]^2$ such that $T(x,y)=0$.
	\end{itemize}
	A strict t-norm is cancellative and a nilpotent t-norm is conditionally cancellative.
\end{definition}

An important field of study concerns the characterization of the structure of different subclasses of triangular norms. One of the main results in this area is the ordinal-sum representation in terms of Archimedean components (see \cite[Theorem 5.1]{Klement2000}), as it allows the study of continuous t-norms to be locally reduced to the Archimedean case in numerous problems. In turn, continuous Archimedean t-norms admit a representation in terms of an additive (or multiplicative) generator, which reduces their expression to an additive (respectively, multiplicative) form that is usually easier to manipulate. In this work, we use the additive generator representation to derive the completion results.
\begin{theorem}[\textbf{\cite[Theorem 5.1]{Klement2000}}]\label{th:representation_archimedeantnorms}
	For a function $T:[0,1]^2 \to [0,1]$ the following statements are equivalent:
	\begin{enumerate}[label=(\roman*)]
		\item $T$ is a continuous and Archimedean t-norm.
		\item $T$ has a continuous additive generator, i.e., there exists a continuous, strictly decreasing function $t:[0,1]\to [0,+\infty]$ with $t(1)=0$, which is uniquely determined up to a positive multiplicative constant, such that
        \begin{equation*}
		T(x,y)=t^{(-1)}(t(x)+t(y)), \quad \text{for all } x,y \in [0,1],
        \end{equation*}
		where $t^{(-1)}$ is the pseudo-inverse of $t$ given by $t^{(-1)}(x)=t^{-1}(\min\{t(0),x\})$ for all $x \in [0,+\infty]$.
	\end{enumerate}
\end{theorem}
Notice that from Theorem~\ref{th:representation_archimedeantnorms} we deduce that the difference between a strict and a nilpotent Archimedean t-norm in terms of an additive generator is that the value $t(0)$ is infinite or finite, respectively. Accordingly, strict t-norms are positive, whereas nilpotent t-norms may vanish for non-zero points. For nilpotent t-norms, this zero region can be described through their induced negation.
Let us recall that a fuzzy negation is a decreasing function $N:[0,1]\to[0,1]$ satisfying $N(0)=1$ and $N(1)=0$, thus generalizing the classical negation to truth degrees in $[0,1]$. It is called \emph{strong} if it is an involution, i.e.,  $N(N(x))=x$ for all $x\in[0,1]$.
\begin{definition}[{\textbf{\cite{Jenei1998}}}]\label{def:induced_negation}
Let $T$ be a nilpotent t-norm. The \emph{induced negation} of $T$ is the function $N_T:[0,1]\to[0,1]$ defined by
\begin{equation*}
N_T(x)=\max\{y\in[0,1]\mid T(x,y)=0\},
\quad \text{for all }x\in[0,1].
\end{equation*}
\end{definition}
The function $N_T$ is a strong fuzzy negation and determines the zero region of $T$ through
\begin{equation}
T(x,y)=0
\quad\Longleftrightarrow\quad
y\leq N_T(x).
\end{equation}
Moreover, the induced negation can be expressed directly in terms of an additive generator.
\begin{remark}[\textbf{\cite[Remark 8]{Fernandez-PeraltaSub}}]\label{remark:induced_negation_generator}
Let $T$ be a nilpotent t-norm, $N_T$ its induced negation and $t$ an additive generator of $T$. Then,
\begin{equation}
N_T(x)=t^{-1}(t(0)-t(x)),
\quad \text{for all }x\in[0,1].
\end{equation}
\end{remark}

To illustrate the previous definitions and results, we now recall some
well-known t-norms that will be used as running examples throughout the paper.
All the examples have been retrieved from the book \cite{Klement2000}. Since
additive generators are unique up to a positive multiplicative constant, for convenience we display in each case the additive generator $t_a$ satisfying
$t_a(a)=1$ for a fixed but arbitrary $a\in(0,1)$.
\begin{example}\label{ex:tnorms}~
\begin{itemize}
\item The \emph{minimum} t-norm \TM, given by $\TM(x,y)=\min\{x,y\}$ for all
$x,y\in[0,1]$, is continuous and idempotent, hence it is not Archimedean and, by
Theorem~\ref{th:representation_archimedeantnorms}, it admits no continuous
additive generator.
\item The \emph{product} t-norm \TP, given by $\TP(x,y)=xy$ for all
$x,y\in[0,1]$, is continuous, Archimedean and strict, with additive generator
$$
t_a(x)=\log_a x=\frac{\ln x}{\ln a} \quad \text{for all } x\in(0,1],
\qquad t_a(0)=+\infty.
$$
\item The \emph{\L ukasiewicz} t-norm \TLK, given by
$\TLK(x,y)=\max\{x+y-1,0\}$ for all $x,y\in[0,1]$, is continuous, Archimedean
and nilpotent, with additive generator
$$
t_a(x)=\frac{1-x}{1-a} \quad \text{for all } x\in[0,1],
\qquad t_a(0)=\frac{1}{1-a}<+\infty,
$$
and induced negation the standard negation $N_{\TLK}(x)=1-x$ for all
$x\in[0,1]$.
\item The \emph{Acz\'el--Alsina} family $\{T^{AA}_r\}_{r>0}$, given by
$$
T^{AA}_r(x,y)=\exp\left(-\left((-\ln x)^r+(-\ln y)^r\right)^{1/r}\right),
\quad \text{for all } x,y\in[0,1],
$$
with the convention $-\ln 0=+\infty$, consists of continuous, Archimedean and
strict t-norms with additive generator $t_a(x)=(\log_a x)^r$ for all
$x\in(0,1]$ and $t_a(0)=+\infty$, so that
$$
T^{AA}_r(x,y)=a^{\left((\log_a x)^r+(\log_a y)^r\right)^{1/r}},
\quad \text{for all } x,y\in(0,1].
$$
For $r=1$ we recover the product t-norm \TP.
\item The \emph{Yager} family $\{T^{Y}_{\lambda}\}_{\lambda>0}$, given by
$$
T^{Y}_{\lambda}(x,y)=\max\left\{1-\left((1-x)^{\lambda}+(1-y)^{\lambda}\right)^{1/\lambda},\,0\right\},
\quad \text{for all } x,y\in[0,1],
$$
consists of continuous, Archimedean and nilpotent t-norms with additive
generator
$$
t_a(x)=\left(\frac{1-x}{1-a}\right)^{\lambda} \quad \text{for all } x\in[0,1],
\qquad t_a(0)=\frac{1}{(1-a)^{\lambda}}<+\infty.
$$
By Remark~\ref{remark:induced_negation_generator}, its induced negation is
$$
N_{T^{Y}_{\lambda}}(x)=1-\left(1-(1-x)^{\lambda}\right)^{1/\lambda},
\quad \text{for all } x\in[0,1].
$$
For $\lambda=1$ we recover the \L ukasiewicz t-norm \TLK.
\end{itemize}
\end{example}

The definition of t-norm takes the entire region $[0,1]^2$ as its domain. In the context of the completion problem, this domain is replaced by a set $A\subseteq[0,1]^2$. Consequently, we find the concept of pre-t-norm, as that binary operator defined on $A$ that satisfies properties compatible with those of t-norms.

\begin{definition}[{\textbf{\cite[Definition 34]{Fernandez-Peralta2022}}}]\label{def:TA}
A function $T_A:A \to [0,1]$ is called a \emph{pre-t-norm} in $A \subseteq [0,1]^2$ if it satisfies the following conditions:
	$$T_A(x,y)=T_A(y,x), \quad \text{for all } (x,y), (y,x) \in A,$$
	$$T_A(x,T_A(y,z))=T_A(T_A(x,y),z), \quad \text{for all } (y,z), (x,T_A(y,z)), (x,y), (T_A(x,y),z)\in A,$$
	$$T_A(x,y) \leq T_A(x,z), \quad \text{for all } (x,y),(x,z) \in A \text{ such that } y \leq z,$$
	$$T_A(x,z) \leq T_A(y,z), \quad \text{for all } (x,z),(y,z) \in A \text{ such that } x \leq y,$$
	$$T_A(1,x)=x \quad \text{ for all } (1,x) \in A \text{ and } T_A(x,1)=x \quad \text{ for all } (x,1) \in A.$$
\end{definition}

The completion problem consists in deciding whether such a partially defined
operator is the restriction of a t-norm and, if so, in describing all the
operators that extend it.

\begin{definition}[{\textbf{\cite[Definition 33]{Fernandez-Peralta2022}}}]\label{def:completion}
Let $T_A:A \to [0,1]$ be a function defined on $A \subseteq [0,1]^2$. We say that
$T_A$ \emph{can be completed to a t-norm} if there exists a t-norm
$T:[0,1]^2 \to [0,1]$ such that $T_A(x,y)=T(x,y)$ for all $(x,y) \in A$. In this
case, $T$ is said to be a \emph{completion} of $T_A$, and a \emph{continuous
completion} whenever $T$ is continuous.
\end{definition}

The conditions in Definition~\ref{def:TA} are exactly the properties in
Definition~\ref{def:tnorm} restricted to those points of $A$ for which all the
involved values are defined. Consequently, they are necessary for the completion
problem to admit a solution.

\begin{proposition}[{\textbf{\cite[Proposition 35]{Fernandez-Peralta2022}}}]\label{prop:completable_implies_pretnorm}
Let $T_A:A \to [0,1]$ be a function defined on $A \subseteq [0,1]^2$. If $T_A$
can be completed to a t-norm $T$, then $T_A$ is a pre-t-norm.
\end{proposition}

However, these conditions are not sufficient in general, since there exist pre-t-norms
that cannot be completed to any continuous t-norm \cite{Fernandez-PeraltaSub}.
Moreover, as only the properties in Definition~\ref{def:tnorm} restricted to $A$
are required, a pre-t-norm need not be bounded by the minimum, although this
holds whenever $A$ contains the boundary of the unit square.

\begin{proposition}[{\textbf{\cite[Proposition 37]{Fernandez-Peralta2022}}}]\label{prop:bounded_by_min}
Let $T_A:A \to [0,1]$ be a pre-t-norm defined on $A \subseteq [0,1]^2$ such that
$[0,1]^2 \setminus [0,1)^2 \subseteq A$. Then $T_A(x,y) \leq \min\{x,y\}$ for all
$(x,y) \in A$.
\end{proposition}

We also adapt the cancellation properties in
Definition~\ref{def:properties_tnorm} to this setting. These two properties play
a central role in the completion problem, since they are inherited by any
completion and, for suitable regions $A$, they allow us to restrict the class in
which such a completion can be found, distinguishing between strict and nilpotent
t-norms.

\begin{definition}[{\textbf{\cite[Definition 36]{Fernandez-Peralta2022}}}]\label{def:cancellative} Let $T_A:A \to [0,1]$  be a pre-t-norm defined on $A\subseteq [0,1]^2$, then $T_A$ is called:
	\begin{enumerate}[label=(\roman*)]
		\item \emph{cancellative} if $T_A(x,y)=T_A(x,z)$  implies $y=z$ for all $(x,y),(x,z) \in A$, $x>0$ and  $T_A(y,x)=T_A(z,x)$ implies $y=z$ for all $(y,x),(z,x) \in A$, $x>0$.
		\item \emph{conditionally cancellative} if $T_A(x,y)=T_A(x,z)>0$ implies $y=z$ for all $(x,y),(x,z) \in A$ and $T_A(y,x)=T_A(z,x)>0$ implies $y=z$ for all $(y,x),(z,x) \in A$.
	\end{enumerate}
\end{definition}

Clearly, enlarging the region on which the operator is known can only
restrict the family of admissible completions. The next result makes this
monotone behaviour precise and will allow us to combine information coming from
different regions.

\begin{proposition}[{\textbf{\cite[Proposition 5.32]{Fernandez-Peralta2023Thesis}}}]\label{prop:EquivalenceCompletionsTATB}
	Let $T_A: A \to [0,1]$, $T_B: B \to [0,1]$ be two continuous pre-t-norms defined on $A \subsetneq B \subsetneq [0,1]^2$ such that $\TA(x,y)=\TB(x,y)$ for all $(x,y) \in A$.
	\begin{enumerate}[label=(\roman*)]
		\item If $T$ is a continuous completion of $\TB$ then $T$ is also a continuous completion of $\TA$.
		\item If \TA has a unique continuous completion $T$ and \TB can be completed, then $T$ is the unique continuous completion of \TB.
	\end{enumerate} 
\end{proposition}

Many regions of the unit square have been considered for this problem
(see \cite{Fernandez-Peralta2023Thesis} for a comprehensive overview). Among
them, we recall the result concerning the region lying above a level curve of the
operator, in which the pre-t-norm determines a piece of an additive generator on
the interval $[a,1]$, unique once its value at $a$ is fixed.

\begin{theorem}[{\textbf{\cite[Theorem 18]{Fernandez-PeraltaSub}}}]\label{thm:generatorabovelevelcurve}
	Let $a \in (0,1)$ and $l:[a,1] \to [a,1]$ be a continuous, strictly decreasing function such that $l^{-1}=l$, $l(a)=1$, $l(1)=a$. Let $\TA: A \to [0,1]$ be a continuous, cancellative pre-t-norm defined on $A=\{(x,y) \in [a,1]^2 \mid y \geq l(x)\}$ such that $\TA(x,l(x))=a $ for all $x \in [a,1]$. There exists a unique continuous, strictly decreasing function $t_a:[a,1] \to [0,1]$ with $t_a(a)=1$, $t_a(1)=0$,
	\begin{equation*}
	    \TA(x,y)=t_a^{-1} (t_a(x)+t_a(y)), \quad \text{for all } (x,y) \in A.
	\end{equation*}
	In this case, we say that $t_a$ is the \emph{additive generator} of $T_A$ with $t_a(a)=1$.
\end{theorem}

As an example, consider $a \in (0,1)$, $l(x)=1+a-x$ for all $x \in [a,1]$ and
the pre-t-norm $\TA(x,y)=x+y-1$ defined on
$A=\{(x,y) \in [a,1]^2 \mid x+y \geq 1+a\}$, which is the restriction of \TLK{}
to the region above its level curve of value $a$. In this case
Theorem~\ref{thm:generatorabovelevelcurve} yields $t_a(x)=\frac{1-x}{1-a}$ for
all $x \in [a,1]$, that is, the additive generator of \TLK{} with $t_a(a)=1$
restricted to $[a,1]$. We refer the reader to \cite{Fernandez-PeraltaSub} for more examples.

Note that this region does not determine a t-norm by itself, since the
operator remains undetermined below the level curve. Nonetheless, since an
additive generator is unique up to a positive multiplicative constant, the piece it provides is the restriction to $[a,1]$ of the generator of every
compatible completion, once the normalization $t_a(a)=1$ is fixed. This makes it a useful and powerful starting point for the study of other regions containing it, and for the completion problem in general.

\section{The completion problem for pre-t-norms known on a parametric curve}\label{sec:parametric}

In this section, we formulate the completion problem for pre-t-norms whose domain is the image of a parametric curve, restricting our attention to continuous and Archimedean completions. Let $\gamma : I \subseteq [0,1] \to [0,1]^2$ be a parametric curve given by $\gamma(\theta)= (x(\theta),y(\theta))$ for all $\theta \in I$, and let
$T_{\gamma} : \gamma(I) \to [0,1]$ be a continuous pre-t-norm defined on $\gamma(I)$. By Theorem~\ref{th:representation_archimedeantnorms}, finding a continuous and Archimedean completion of $T_{\gamma}$ amounts to determining an additive generator $t$ satisfying the functional equation
\begin{equation}\label{eq:feq:1}
t^{(-1)}(t(x(\theta)) + t(y(\theta))) = T_{\gamma}(x(\theta),y(\theta)),
\quad \text{for all } \theta \in I.
\end{equation}
To solve the problem, one may focus on two key points:
\begin{enumerate}[label=(\roman*)]
	\item The curve $\gamma$ has to provide enough information to solve the equation.
	\item We have to further study the additional conditions that the function $T_{\gamma}$ has to satisfy so that it is compatible with the structure of continuous and Archimedean t-norms.
\end{enumerate}
Further, if we consider different parametric curves, then for each we have to consider Eq.~(\ref{eq:feq:1}), yielding a system of functional equations \cite{Darsow1983}. Although different regions of this type have been considered (see \cite{Fernandez-Peralta2019,Fernandez-PeraltaSub}), the only regions that are solved by providing an explicit reconstruction of an additive generator of every possible completion are when $T_{\gamma}$ corresponds either to the diagonal or to a horizontal (resp.\ vertical) section, i.e., they solve Eq.~(\ref{eq:feq:1}) under the following conditions:
\begin{itemize}
    \item If $I=[0,1]$, $\gamma(\theta)=(\theta,\theta)$ for all $\theta \in [0,1]$ and $T_{\gamma}(\theta,\theta) = \delta(\theta)$ is a function such that there exists $d\in[0,1)$ with $\delta(x)=0$ for all $x\in[0,d]$, $\delta$ strictly increasing on $[d,1]$, $\delta(1)=1$ and $\delta(x)<x$ for all $x\in(0,1)$.
    \item If $I=[0,1]$, $\alpha \in (0,1)$, $\gamma(\theta) = (\theta,\alpha)$ for all $\theta \in [0,1]$ and $T_{\gamma}(\theta,\alpha)=h_{\alpha}(\theta)$ is a function such that there exists $c\in[0,1)$ with $h_{\alpha}(x)=0$ for all $x\in[0,c]$, $h_{\alpha}$ strictly increasing on $[c,1]$, $h_{\alpha}(1)=\alpha$ and $h_{\alpha}(x)<x$ for all $x\in(0,1]$.
\end{itemize}
In particular, the solutions correspond to the following theorems.
\begin{theorem}[\textbf{\cite[Theorem~3.8.2]{Alsina2006}}] \label{theorem:alsina:diagonal}
Let $\delta\colon [0,1]\to[0,1]$ be a function such that there exists $d\in[0,1)$ with $\delta(x)=0$ for all $x\in[0,d]$, $\delta$ strictly increasing on $[d,1]$, $\delta(1)=1$ and $\delta(x)<x$ for all $x\in(0,1)$. We call $\delta$ a \emph{diagonal}. Then there are infinitely many distinct continuous and Archimedean t-norms $T$ for which $T(\cdot,\cdot)=\delta$. The additive generators $t$ of all such $T$ satisfying $t(x_0)=1$ depend on an arbitrary strictly decreasing continuous function $\overline{t}\colon[x_1,x_0]\to[1,2]$ with $\overline{t}(x_0)=1$ and $\overline{t}(x_1)=2$, and are given by
\begin{equation}\label{eq:schroder}
    t(x) = 2^n\, \overline{t}\!\left(\delta^{-n}(x)\right), \qquad x_{n+1} < x \leq x_n,
\end{equation}
where $x_0$ is chosen arbitrarily in $(d,1)$ if $d=0$, or as the unique point in $(d,1)$ such that $\delta(x_0)=d$ if $d>0$, and $x_n=\delta^n(x_0)$ for all $n\in\mathbb{Z}$, together with $t(1)=0$ and $t(0)=+\infty$ if $d=0$.
\end{theorem}
\begin{theorem}[\textbf{\cite[Theorem~3.8.4]{Alsina2006}}]\label{theorem:alsina:horizontal}
Let $\alpha \in (0,1)$ and let $h_{\alpha}\colon [0,1]\to[0,\alpha]$ be a function such that there exists $c\in[0,1)$ with $h_{\alpha}(x)=0$ for all $x\in[0,c]$, $h_{\alpha}$ strictly increasing on $[c,1]$, $h_{\alpha}(1)=\alpha$ and $h_{\alpha}(x)<x$ for all $x\in(0,1]$. We call $h_\alpha$ an \emph{$\alpha$-section}. Then there are infinitely many distinct continuous and Archimedean t-norms $T$ for which $T(\cdot,\alpha)=h_{\alpha}$. The additive generators $t$ of all such $T$ satisfying $t(\alpha)=1$ depend on an arbitrary strictly decreasing continuous function $\overline{t}\colon[\alpha,1]\to[0,1]$ with $\overline{t}(1)=0$ and $\overline{t}(\alpha)=1$, and are given by
\begin{equation}\label{eq:abel}
    t(x) = \begin{cases}
        n + \overline{t}\!\left(h_{\alpha}^{-n}(x)\right) & \text{if } y_n \leq x < y_{n-1}, \quad n \geq 1, \\
        \overline{t}(x) & \text{if } \alpha \leq x \leq 1,
    \end{cases}
\end{equation}
where $y_0=\alpha$ and $y_n=h_{\alpha}^n(\alpha)$ for all $n\geq 0$, together with $t(0)=+\infty$ if $c=0$, or $t(0)=N+1$ where $N = \inf\{n \geq 0 : y_n = 0\}$ if $c>0$.
\end{theorem}

Going through the corresponding proofs, one realizes that the two problems are solved in a very similar way, since they can be treated iteratively: once an initial piece of the generator is fixed, two of the three terms in Eq.~(\ref{eq:feq:1}) are determined, and the generator can be extended step by step to the whole interval. Further, notice that in both cases the curve $\gamma$ has to fulfill additional conditions in order for $T_{\gamma}$ to be the diagonal (resp.\ horizontal section) of a continuous and Archimedean t-norm.

In this paper, our aim is to further study this problem and determine under which conditions Eq.~(\ref{eq:feq:1}) can be solved beyond the particular cases of the diagonal and the horizontal section. Although it is clear that the problem of solving Eq.~(\ref{eq:feq:1}) can be stated in general, it is hard to know beforehand which curves $\gamma$ will have a nice enough structure to ensure an explicit reconstruction of the t-norm. As a first step, in this paper we focus on studying Eq.~(\ref{eq:feq:1}) for curves $\gamma$ that fix one of the components, i.e., $\gamma(\theta) = (\theta, p(\theta))$ where $p\colon[0,1] \to [0,1]$ is a continuous increasing function. We then assume that the known part of the t-norm corresponds to $f(x)=T(x,p(x))$ for all $x \in [0,1]$. This assumption is reasonable since it generalizes both the diagonal and the horizontal section, which are recovered by taking $p(x)=x$, and $p(x)=\alpha$ for some fixed $\alpha \in (0,1)$, respectively. Moreover, it provides boundary conditions from $f(0)=T(0,p(0))=0$ and $f(1)=T(1,p(1))=p(1)$, and allows us to leverage the monotonicity of both $f$ and $p$ in order to study under which conditions an explicit reconstruction of the additive generator of $T$ is possible.

In this direction, we first identify when the restriction to such a curve forces
the completion to be Archimedean. Specifically, we prove that if a pre-t-norm $T_A$ is
defined on a region containing the curve $\{(x,p(x)) : x \in [0,1]\}$ and satisfies
$T_A(x,p(x)) < \min\{x,p(x)\}$ for all $x \in (0,1)$, then $p(x) \in (0,1)$ for all
$x \in (0,1)$ and every continuous completion of $T_A$ is Archimedean.

\begin{proposition}\label{prop:archimedean} 
Let $p\colon [0,1] \to [0,1]$ be a continuous increasing function and let $T_A \colon A \to [0,1]$ be a continuous pre-t-norm defined on a region $A \subset [0,1]^2$ containing the curve $\{(x,p(x)) : x \in [0,1]\}$. If $T_A$ is completable to a continuous t-norm $T$ and $T_A(x,p(x)) < \min\{x,p(x)\}$ for all $x \in (0,1)$, then $p(x) \in (0,1)$ for all $x \in (0,1)$ and $T$ is Archimedean. Moreover, if $T_A(x,p(x))>0$ for all $x>0$, then $T$ is strict; otherwise, it is nilpotent.
\end{proposition}

\begin{proof}
First, suppose that $p(x)=1$ for some $x \in (0,1)$. Then $T(x,p(x))=T(x,1)=x$, which contradicts $T_A(x,p(x)) < \min\{x,p(x)\}=x$. The case $p(x)=0$ is analogous, since $T(x,p(x))=T(x,0)=0=p(x)$, again contradicting the strict inequality. Hence $p(x) \in (0,1)$ for all $x \in (0,1)$.

Now suppose that $T$ is not Archimedean. Then, by the characterization of continuous t-norms, there exists $\alpha \in (0,1)$ such that $T(\alpha,\alpha)=\alpha$ and $T(x,y) = \min\{x,y\}$ for all $(x,y) \in ([\alpha,1] \times [0,\alpha]) \cup ([0,\alpha] \times [\alpha,1])$. Since $p(\alpha) \in (0,1)$, we have that $(\alpha, p(\alpha))$ belongs to one of these regions, and thus $T(\alpha,p(\alpha)) = \min\{\alpha,p(\alpha)\}$,
which contradicts $T_A(\alpha,p(\alpha)) < \min\{\alpha,p(\alpha)\}$. Hence $T$ is Archimedean.

Finally, since $T$ is continuous and Archimedean, it is either strict or nilpotent.
If $T$ is nilpotent, its induced negation $N_T$ is a strong negation, hence continuous
with $N_T(0)=1$. As $p$ is increasing with $p(x)\in(0,1)$ for all $x\in(0,1)$, we have
$p(0)<1$, so by continuity there exists $x>0$ with $p(x)\leq N_T(x)$, and therefore
$T_A(x,p(x))=T(x,p(x))=0$. Consequently, if $T_A(x,p(x))>0$ for all $x>0$, then $T$
cannot be nilpotent and is thus strict; and if $T_A(x,p(x))=0$ for some $x>0$, then
$T$ is not positive and is thus nilpotent.
\end{proof}

We now generalize the ideas of the diagonal $\delta$ and the horizontal section $h_{\alpha}$ introduced in Theorems \ref{theorem:alsina:diagonal} and \ref{theorem:alsina:horizontal}
to a broader setting that we call a \emph{$p$-curve}. The conditions we impose on this function are precisely those required for it to represent the section $T(x,p(x))$ of a continuous and Archimedean t-norm, as shown in Proposition~\ref{prop:archimedean}.

\begin{definition}\label{def:pcurve}
Let $p\colon [0,1] \to [0,1]$ be a continuous increasing function with $p(x) \in (0,1)$ for all $x \in (0,1)$. We define a \emph{$p$-curve} as a function $f\colon[0,1] \to [0,p(1)]$ which satisfies the following conditions:
\begin{enumerate}[label=(\roman*)]
\item $f$ is continuous with $f(0)=0$ and $f(1)=p(1)$;
\item There exists $c \in [0,1)$ such that $f(x)=0$ for all $x \in [0,c]$ and $f$ is strictly increasing on $[c,1]$;
\item $f(x) < \min\{x,p(x)\}$ for all $x \in (0,1)$.
\end{enumerate}
In this case, we call $c = \sup\{x \in [0,1] \mid f(x)=0\}$ the \emph{zero threshold} of $f$.
\end{definition}

Indeed, the following proposition shows that the section $f(x)=T(x,p(x))$ of any continuous and Archimedean t-norm $T$ is always a $p$-curve.

\begin{proposition}\label{prop:necessity_pcurve}
    Let $p\colon[0,1]\to[0,1]$ be a continuous increasing function with
    $p(x)\in(0,1)$ for all $x\in(0,1)$. If $T$ is a continuous, Archimedean t-norm
    then $f(x)=T(x,p(x))$ is a $p$-curve with zero threshold $c$. Moreover,
    $c=0$ if and only if $T$ is strict.
\end{proposition}
\begin{proof}
Since $T$ is a t-norm, we have $f(0) = T(0, p(0)) = 0$ and $f(1) = T(1, p(1)) = p(1)$. Moreover, since $T$ is increasing in each component and $p$ is increasing, $f$ is increasing as well. Since $p(1)>0$, by continuity there exists $c \in [0,1)$ such that $f(x)=0$ for all $x \in [0,c]$ and $f(x)>0$ for all $x \in (c,1]$. Furthermore, since $p$ is increasing and $T$ is conditionally cancellative on $(0,1)^2$, for $c \leq x<y \leq 1$ we have
$f(x)=T(x,p(x)) \leq T(x,p(y)) < T(y,p(y))=f(y)$, where the strict inequality holds
because $f(y)>0$. We conclude that $f$ is strictly increasing on $[c,1]$. The condition $f(x) < \min\{x,p(x)\}$ for all $x \in (0,1)$ follows directly from the fact that $T$ is a continuous and Archimedean t-norm, which satisfies $T(x,y) < \min\{x,y\}$ for all $(x,y) \in (0,1)^2$. Finally, $c=0$ if and only if $f(x)>0$ for all $x\in(0,1]$, that is, if and only if
$T(x,p(x))>0$ for all $x\in(0,1]$. If $T$ is strict then it is positive, and since
$p(x)>0$ for all $x\in(0,1]$, this holds. Conversely, if $T$ is nilpotent then
$N_T$ is a strong negation and, as $p(0)<1$, by continuity there exists $x>0$ with
$p(x)\leq N_T(x)$, so $f(x)=0$ and $c>0$.
\end{proof}

\noindent In the particular cases of the diagonal and the horizontal section $h_{\alpha}$, the zero threshold $c$ is given explicitly by $c=t^{(-1)}(\frac{1}{2}t(0))$ and $c=t^{(-1)}(t(0)-t(\alpha))$, respectively, where $t$ is an additive generator of $T$. In the more general framework, $c$ still exists and is unique but can only be characterized implicitly as $c = \sup\{x \in [0,1] \mid f(x)=0\}$, which is the solution of $0 = t^{(-1)}(t(x)+t(p(x)))$.

With the notion of a $p$-curve at hand, we can now return to Eq.~(\ref{eq:feq:1}) and study the conditions under which it can be solved. If $f$ is a $p$-curve, then Eq.~(\ref{eq:feq:1}) reduces to the following functional equation
\begin{equation}\label{eq:completion_genversion}
t(f(x)) = t(x)+t(p(x)), \quad c \leq x \leq 1.
\end{equation}
As in the diagonal and horizontal section cases, we expect the reconstruction of the additive generator to rely on an iterative procedure. As a first step in this direction, we study the behavior of the iterates of $f$, in analogy with the sequences $\{x_n\}$ and $\{y_n\}$ appearing in Theorems~\ref{theorem:alsina:diagonal} and~\ref{theorem:alsina:horizontal}, respectively. Specifically, we prove that the powers of any value $a \in (0,p(1)] \cap (0,1)$ with respect to $f$ generate a partition of the unit interval, without imposing any further conditions on $f$ beyond those in Definition~\ref{def:pcurve}.

\begin{proposition}\label{prop:properties-f-with-pcurve}
    Let $f$ be a $p$-curve with zero threshold $c$ and let $a \in (0,p(1)] \cap (0,1)$. Define $f_c\colon[c,1] \to [0,p(1)]$ as the restriction $f_c = f\mid_{[c,1]}$. Then the following statements hold:
    \begin{enumerate}[label=(\roman*)]
    \item $f_c$ is a continuous strictly increasing function with $f_c(c)=0$ and $f_c(1)=p(1)$. In particular, $f_c$ is a homeomorphism onto its image, and its inverse $f_c^{-1}\colon[0,p(1)] \to [c,1]$ is continuous and strictly increasing.
    \item Denote by $f^{n}=\overbrace{f \circ \cdots \circ f}^{n}$, with
    $f^{0}=\mathrm{id}$, the iterates of $f$, which are well defined on $[0,1]$ since
    $f([0,1]) \subseteq [0,1]$. The sequence $x_n = f^{n}(a)$ is decreasing and converges to $0$, and it is strictly
decreasing as long as $x_n > 0$. Moreover, if
    $$
    n_f = \inf\{n \in \mathbb{N} \mid f^{n}(a) = 0\} \in \mathbb{N} \cup \{+\infty\},
    $$
    then the intervals $[f^n(a),f^{n-1}(a))$, $1 \leq n \leq n_f$, and $[a,1]$ are
    pairwise disjoint and
    $$
        \left(\bigcup_{n=1}^{n_f} [f^n(a),f^{n-1}(a))\right)\cup[a,1]=
        \begin{cases} [0,1], & \text{if } n_f<+\infty,\\ (0,1], & \text{if } n_f=+\infty.\end{cases}
    $$
    Furthermore, $n_f<+\infty$ if and only if $c>0$, in which case
    $f^{n_f-1}(a) \leq c$ and $f^{n}(a) > c$ for all $n \in \{0,\dots,n_f-2\}$.
    \item For each $1 \leq n \leq n_f$, the iterate
    $f_c^{-n}=\overbrace{f_c^{-1} \circ \cdots \circ f_c^{-1}}^n$, with
    $f_c^{0}=\mathrm{id}$, is well defined, continuous and strictly increasing on
    $[f^n(a),f^{n-1}(a))$, and, for $n < n_f$, satisfies $f_c^{-n}=(f^n)^{-1}$ there. Moreover,
    $$
    f_c^{-1}\bigl([f^n(a),f^{n-1}(a))\bigr) \subseteq [f^{n-1}(a),f^{n-2}(a)),
    \quad 2 \leq n \leq n_f,
    $$
    and $f_c^{-1}([f(a),a)) \subseteq [a,1]$, so that
    $f_c^{-n}([f^n(a),f^{n-1}(a))) \subseteq [a,1]$.
    \end{enumerate}
\end{proposition}
\begin{proof}
\begin{enumerate}[label=(\roman*)]
\item Follows directly from Definition~\ref{def:pcurve}.

\item Let $x_n = f^n(a)$ for all $n \in \mathbb{N}$. Since $f(x) < x$ for all
$x \in (0,1]$ by (iii)-Definition~\ref{def:pcurve}, the sequence is strictly decreasing
as long as $x_n > 0$, and equals $0$ thereafter; in particular it is decreasing. If
$f^n(a) = 0$ for some $n \in \mathbb{N}$, then $x_k = 0$ for all $k \geq n$ and
trivially $x_n \to 0$. Otherwise, $f^n(a) > 0$ for all $n \in \mathbb{N}$, the sequence
is strictly decreasing and bounded below, so its limit $L \geq 0$ exists. If $L > 0$,
then by continuity $f(L) = L$, contradicting (iii)-Definition~\ref{def:pcurve}. Hence,
$L = 0$.

For the partition equality, note that the intervals $[f^n(a), f^{n-1}(a))$ are pairwise
disjoint and cover $[f^{n_f}(a), a) = [0, a)$ when $n_f < +\infty$, and $(0,a)$ in
the limit when $n_f = +\infty$. Adding $[a,1]$ gives $[0,1]$ and $(0,1]$,
respectively.

Finally, if $c=0$ then $f(x)>0$ for all $x>0$, so $f^n(a)>0$ for every $n$ and
$n_f=+\infty$. If $c>0$, then $x_n \to 0$ gives $x_n \leq c$ for some $n$, whence
$x_{n+1}=0$ and $n_f<+\infty$. In the latter case $f^{n_f}(a)=0$ means
$f^{n_f-1}(a) \leq c$, and $f^n(a)>c$ for all $n \in \{0,\dots,n_f-2\}$ by minimality
of $n_f$.

\item We first check that $f_c^{-1}$ can be applied on each interval, and then
we iterate. Since $f^{n-1}(a) \leq a \leq p(1)$, the interval $[f^n(a),f^{n-1}(a))$ is
contained in $[0,p(1)]=\Dom(f_c^{-1})$, so $f_c^{-1}$ can be applied.

For $n=1$, as $a \leq p(1)$ we have $f_c^{-1}(a) \leq 1$ and hence
$f_c^{-1}([f(a),a))  \subseteq [a,1]$.

For $2 \leq n \leq n_f$, the map $f_c$ sends $[f^{n-1}(a),f^{n-2}(a))$ onto $[f^{n}(a),f^{n-1}(a))$ when $n \leq n_f-1$, and sends $[c,f^{n_f-2}(a))$ onto $[0,f^{n_f-1}(a))$ when $n=n_f$. In both cases, $f_c^{-1}([f^n(a),f^{n-1}(a)))$ is contained in $[f^{n-1}(a),f^{n-2}(a))$. Iterating $n$ times we obtain
$f_c^{-n}([f^n(a),f^{n-1}(a))) \subseteq [a,1]$. Continuity, strict monotonicity and the identity $f_c^{-n}=(f^n)^{-1}$ are derived from (i).
\end{enumerate}
\end{proof}

What remains to be addressed is under which conditions a recursive argument can be used to completely reconstruct the additive generator from Eq.~(\ref{eq:completion_genversion}). Since any two of the three values involved determine the third, this equation can be read in several ways. Among them, computing $t(f(x))$ from $t(x)$ and $t(p(x))$ is the only one in which the value obtained always corresponds to a point below the two used, since $f(x)<\min\{x,p(x)\}$ for all $x \in (0,1)$, whereas for the remaining alternatives this depends on the position of $p$ with respect to the diagonal, and so does the part of the unit interval that the corresponding recursion is able to reach. We adopt the first reading throughout, and we will see that it already suffices to describe all the completions and combines well with the information given above a level curve. Since each of its steps gives a new value at a point below those already known, we have to fix an initial piece of the generator on an interval $[a,1]$ and propagate its values to $[0,a)$, as in Theorems~\ref{theorem:alsina:diagonal} and~\ref{theorem:alsina:horizontal}. On the other hand, nothing new is obtained where the three arguments already lie in $[a,1]$, so there the equation becomes a condition on the initial piece. This problem is studied in the subsequent sections, first in the particular case where $p$ is bounded away from zero, and then under a compatibility condition between $p$ and the iterates of $f$.

\section{Continuous completion when $p$ is bounded away from zero}\label{sec:constant}

In this section, we study the case where $p(0) > 0$. Since $p$ is increasing, we
have $p(x) \geq p(0) > 0$ for all $x \in [0,1]$, so we may take any
$a \in (0,p(0)]$ as the base of the recursive construction. Under this assumption,
Eq.~\eqref{eq:completion_genversion} has a particularly nice structure. Indeed, if the
additive generator $t$ is known on the interval $[a,1]$, then for any
$x \in [f(a),a)$ two of the three terms in Eq.~\eqref{eq:completion_genversion} are
already determined, since both $f_c^{-1}(x) \in [a,1]$ and $p(f_c^{-1}(x)) \geq a$ lie
in the known region. This allows us to build a recursive argument that extends the
generator step by step to the whole interval. Moreover, we prove that the
resulting expression describes not only some, but all the additive generators of the
t-norms whose $p$-curve section equals $f$.

\begin{theorem}\label{thm:completion:1}
Let $f$ be a $p$-curve with $p(0)>0$ and let $a \in (0,p(0)]$. Then, $T$ is a continuous t-norm with $T(x,p(x))=f(x)$ for all $x \in [0,1]$ if and only if $T$ is an Archimedean t-norm with additive generator
\begin{equation}\label{gen:completion:1}
t(x) = 
\begin{cases}
+\infty, & \text{if } x=0 \text{ and } n_f = +\infty,\\
\displaystyle t_a\!\bigl(f_c^{-n}(x)\bigr) 
+ \sum_{i=1}^{n} t_a\!\bigl(p(f_c^{-i}(x))\bigr),
& \text{if } f^{n}(a) \le x < f^{n-1}(a), \qquad 1 \le n \le n_f, \\[1em]
t_a(x),
& \text{if } a \le x \le 1,
\end{cases}
\end{equation}
where $t_a : [a,1] \to [0,1]$ is a continuous, strictly decreasing function with $t_a(1)=0$, $t_a(a)=1$ and
\begin{equation}\label{eq:consistency_bb}
t_a(f(x)) = t_a(x) + t_a(p(x)), \quad \text{for all } x \in [f_c^{-1}(a),1].
\end{equation}
\end{theorem}
\begin{proof}
First, let $T$ be a continuous t-norm with $T(x,p(x))=f(x)$ for all $x \in [0,1]$. By Proposition \ref{prop:archimedean} we know that $T$ is Archimedean, so there exists an additive generator $t$ such that $t^{(-1)}(t(x)+t(y)) = T(x,y)$. Without loss of generality, we may assume that $t(a)=1$, otherwise we replace $t$ with the additive generator $t(x)/t(a)$. We denote by $t_a$ the restriction of $t$ to $[a,1]$, which is continuous, strictly decreasing and satisfies $t_a(1)=0$ and $t_a(a)=1$.

Let us first check that $t_a$ satisfies Eq.~(\ref{eq:consistency_bb}). If $x \in [f_c^{-1}(a),1]$, then $f(x) \geq a$ and $p(x) \geq p(0) \geq a$, so the three values $t(x)$, $t(p(x))$ and $t(f(x))$ are given by $t_a$. Since $t$ generates $T$ and $f(x)=T(x,p(x))>0$, we obtain $t_a(f(x))=t_a(x)+t_a(p(x))$.

Let $x \in (0,a)$, then there exists $1 \leq n \leq n_f$ such that $x \in [f^n(a),f^{n-1}(a))$, we will prove that $t$ corresponds to the additive generator in Eq.~(\ref{gen:completion:1}) by induction on $n$.

\begin{itemize}
\item If $n=1$ then $x \in [f(a),a)$, so $f_c^{-1}(x) \in [a,1]$, and
\begin{equation*}
x = f_c \circ f_c^{-1}(x) = T(f_c^{-1}(x),p(f_c^{-1}(x))) = t^{(-1)}(t(f_c^{-1}(x)) + t(p(f_c^{-1}(x)))),
\end{equation*}
\begin{equation*}
t(x) = t_a(f_c^{-1}(x)) + t_a(p(f_c^{-1}(x))).
\end{equation*}

\item We assume that Eq.~(\ref{gen:completion:1}) holds for $n$ and we prove it for $n+1$. 
If $x \in [f^{n+1}(a), f^{n}(a))$ then $f_c^{-1}(x) \in [f^{n}(a), f^{n-1}(a))$, and
\begin{equation*}
x = f_c \circ f_c^{-1}(x) = T(f_c^{-1}(x),p(f_c^{-1}(x)))= t^{(-1)}(t(f_c^{-1}(x)) + t(p(f_c^{-1}(x)))),
\end{equation*}
\begin{eqnarray*}
    t(x) & = & t(f_c^{-1}(x)) + t_a(p(f_c^{-1}(x))) = t_a(f_c^{-n}(f_c^{-1}(x))) + \sum_{i=1}^{n} t_a(p(f_c^{-i}(f_c^{-1}(x)))) + t_a(p(f_c^{-1}(x))) \\
         & = & t_a(f_c^{-(n+1)}(x)) + \sum_{i=1}^{n+1} t_a(p(f_c^{-i}(x))).
\end{eqnarray*}
\end{itemize}
If $x=0$ then $t(0)=+\infty$ whenever $n_f=+\infty$ by definition, and otherwise 
\begin{equation*}
t(0) = t_a(f_c^{-n_f}(0)) + \sum_{i=1}^{n_f} t_a(p(f_c^{-i}(0)))< + \infty.
\end{equation*}

For the reverse implication, we first prove that the function defined in Eq.~(\ref{gen:completion:1}) is an additive generator of a continuous and Archimedean t-norm, i.e., we prove that $t$ is a continuous and strictly decreasing function with ${t(1)=0}$.

\begin{description}
    \item[Monotonicity.] Since $t_a$ is strictly decreasing and $f_c^{-1}$ and $p$ are increasing, $t$ is strictly decreasing on each interval $[f^{n}(a),f^{n-1}(a))$, $1 \leq n \leq n_f$, and on $[a,1]$. Moreover, at each junction the values match, since
    \begin{equation*}
    \lim_{x \to a^-} t(x) = t_a(f_c^{-1}(a)) + t_a(p(f_c^{-1}(a))) = t_a(a) = t(a),
    \end{equation*}
    where we used Eq.~(\ref{eq:consistency_bb}) at $x=f_c^{-1}(a)$, and analogously, for $1 \leq n < n_f$,
    \begin{equation*}
    \lim_{x \to f^{n}(a)^-} t(x) = t_a(f_c^{-1}(a)) + t_a(p(f_c^{-1}(a))) + \sum_{i=1}^{n} t_a(p(f^{n-i}(a))) = t_a(a) + \sum_{i=1}^{n} t_a(p(f^{n-i}(a))) = t(f^{n}(a)).
    \end{equation*}
    As the intervals are consecutive and $t$ is strictly decreasing on each of them, we deduce that $t$ is a strictly decreasing function on $(0,1]$. Finally, if $n_f=+\infty$ then $t(0)=+\infty>t(x)$ for all $x \in (0,1]$, and if $n_f<+\infty$ then $0$ belongs to the interval $[f^{n_f}(a),f^{n_f-1}(a))$, on which $t$ is strictly decreasing. Hence $t$ is strictly decreasing on $[0,1]$.
    \item[Continuity.] Since $t_a$, $f$ and $p$ are continuous functions, $t$ is continuous inside each interval of the definition. As the intervals are closed on the left, each junction point $a$ or $f^{n}(a)$ belongs to the piece lying to its right, so the corresponding right-hand limit is immediate, and the computations above show that the left-hand limits coincide with $t(a)$ and $t(f^{n}(a))$, respectively. Thus, it only remains to study the point $x=0$.
    \begin{itemize}
    \item If $n_f < +\infty$, then $t$ is continuous at $x=0$ since only the right-hand limit exists at that point, and we already justified that the function is continuous inside the interval.
    \item If $n_f = +\infty$, then for every $n \geq 1$ and every $x < f^{n}(a)$ we have
    \begin{equation*}
    t(x) \geq t(f^{n}(a)) = 1 + \sum_{i=1}^{n} t_a(p(f^{n-i}(a))) \geq 1 + n\, t_a(p(a)),
    \end{equation*}
    since $f^{n-i}(a) \leq a$ implies $p(f^{n-i}(a)) \leq p(a)$ and hence $t_a(p(f^{n-i}(a))) \geq t_a(p(a))>0$, where $t_a(p(a))>0$ because $p(a)<1$. Letting $n \to +\infty$ we obtain $\lim_{x \to 0^+} t(x) = +\infty = t(0)$.
    \end{itemize}
    \item[Boundary condition.] $t(1)=t_a(1)=0$ by construction.
\end{description}

Finally, let us prove that if $T$ is a t-norm with additive generator in Eq.~(\ref{gen:completion:1}) then $T(x,p(x)) =f(x)$ for all $x\in [0,1]$. We distinguish three cases according to the position of $x$ in the partition of Proposition~\ref{prop:properties-f-with-pcurve}.

\begin{itemize}
\item If $x \in [0,c]$, which is non-empty only when $c>0$ and hence $n_f<+\infty$, then $f(x)=0$ and $c \geq f^{n_f-1}(a)$. We distinguish two subcases.
\begin{itemize}
\item If $n_f=1$ then $c \in [a,1]$ and
\begin{equation*}
t(x) + t(p(x)) \geq t(c) + t(p(c)) = t_a(c)+t_a(p(c)) = t(0).
\end{equation*}
\item If $n_f \geq 2$ then $c < f^{n_f-2}(a)$ and
\begin{eqnarray*}
t(x)+t(p(x)) & \geq & t(c) + t(p(c)) = t_a(f_c^{-n_f+1}(c)) + \sum_{i=1}^{n_f-1} t_a(p(f_c^{-i}(c))) + t_a(p(c)) \\
& = & t_a(f_c^{-n_f}(0)) + \sum_{i=1}^{n_f} t_a(p(f_c^{-i}(0))) = t(0).
\end{eqnarray*}
\end{itemize}
In either case, we have $T(x,p(x)) = 0 = f(x)$.
\item If $x \in (c,a)$, then $f(x)>0$ and there exists a $1 \leq n \leq n_f$ such that $x \in [f^{n}(a),f^{n-1}(a))$. Note that $n \leq n_f-1$ whenever $n_f<+\infty$, since $x>c \geq f^{n_f-1}(a)$. Then $f(x) \in [f^{n+1}(a),f^{n}(a))$ and we have
\begin{eqnarray*}
t(f(x)) &=& t_a(f_c^{-n-1}(f(x))) + \sum_{i=1}^{n+1} t_a(p(f_c^{-i}(f(x)))) = t_a(f_c^{-n}(x)) + \sum_{i=1}^{n+1} t_a(p(f_c^{-i+1}(x))) \\
& = & t_a(f_c^{-n}(x)) + \sum_{i=1}^n t_a(p(f_c^{-i}(x))) + t_a(p(x)) = t(x) + t(p(x)).
\end{eqnarray*}
\item If $x \in [a,1]$, we distinguish two subcases. If $x \in [a,f_c^{-1}(a))$, then $f(x) \in [f(a),a)$ and, applying the case $n=1$ of Eq.~(\ref{gen:completion:1}),
\begin{equation*}
t(f(x)) = t_a(f_c^{-1}(f(x))) + t_a(p(f_c^{-1}(f(x)))) = t_a(x) + t_a(p(x)) = t(x)+t(p(x)).
\end{equation*}
If $x \in [f_c^{-1}(a),1]$, then $f(x) \geq a$ and $p(x) \geq a$, so the three values are given by $t_a$ and the identity $t(f(x))=t(x)+t(p(x))$ is exactly Eq.~(\ref{eq:consistency_bb}).
\end{itemize}
In the second and third cases $f(x)>0$ and $t(f(x))=t(x)+t(p(x))$, hence $T(x,p(x))=t^{(-1)}(t(x)+t(p(x)))=f(x)$.
\end{proof}

As pointed out in Section~\ref{sec:parametric}, $p$-curves are generalizations of both the diagonal and the $\alpha$-section concepts. In particular, taking $p$ to be a constant function $p(x)=\alpha$ for some fixed $\alpha \in (0,1)$, the $p$-curve reduces to an $\alpha$-section and $p(0)=\alpha>0$, so Theorem~\ref{thm:completion:1} applies with $a=\alpha$. In the following corollary, we show that in this case Theorem~\ref{thm:completion:1} reduces exactly to Theorem~\ref{theorem:alsina:horizontal}.
\begin{corollary}\label{cor:horizontal}
Let $\alpha \in (0,1)$ and let $h_\alpha$ be an $\alpha$-section. Then $T$ is a continuous t-norm with $T(x,\alpha)=h_\alpha(x)$ for all $x\in[0,1]$ if and only if $T$ is an Archimedean t-norm with additive generator satisfying $t(\alpha)=1$ and given by
\begin{equation}\label{eq:gen:horizontal}
t(x) = 
\begin{cases}
+\infty, & \text{if } x=0 \text{ and } n_f = +\infty,\\
\displaystyle t_\alpha\!\bigl(h_\alpha^{-n}(x)\bigr) + n,
& \text{if } h_\alpha^{n}(\alpha) \le x < h_\alpha^{n-1}(\alpha), \quad 1 \le n \le n_f,\\
t_\alpha(x),
& \text{if } \alpha \le x \le 1,
\end{cases}
\end{equation}
where $t_\alpha\colon[\alpha,1]\to[0,1]$ is a continuous strictly decreasing function with $t_\alpha(1)=0$ and $t_\alpha(\alpha)=1$.
\end{corollary}
\begin{proof}
Take $p(x)=\alpha$ for all $x\in[0,1]$, $f=h_\alpha$ and $a=\alpha$ in Theorem~\ref{thm:completion:1}. Since $p$ is constant, $f_c^{-1}(\alpha)=h_\alpha^{-1}(\alpha)=1$, so the interval $[f_c^{-1}(a),1]$ reduces to the single point $1$ and Eq.~(\ref{eq:consistency_bb}) corresponds to
\begin{equation*}
t_\alpha(h_\alpha(1)) = t_\alpha(1)+t_\alpha(\alpha),
\end{equation*}
which holds since $h_\alpha(1)=\alpha$ and $t_\alpha(1)=0$. Hence, the only conditions on $t_\alpha$ are continuity, strict monotonicity and the two boundary values. Moreover, the sum in Eq.~(\ref{gen:completion:1}) reduces to
\begin{equation*}
\sum_{i=1}^{n} t_\alpha\bigl(p(h_\alpha^{-i}(x))\bigr) = n\, t_\alpha(\alpha) = n,
\end{equation*}
which gives Eq.~(\ref{eq:gen:horizontal}).
\end{proof}
Eq.~(\ref{eq:gen:horizontal}) matches the generator formula of Theorem~\ref{theorem:alsina:horizontal} with $\overline{t}=t_\alpha$, and the $\alpha$-section is the degenerate case in which Eq.~(\ref{eq:consistency_bb}) becomes vacuous, since $f_c^{-1}(\alpha)=1$ and hence $t_\alpha$ is entirely free on $[\alpha,1]$. This only happens when $p$ is constant. Otherwise, $f_c^{-1}(a)<1$ and Eq.~(\ref{eq:consistency_bb}) is a non-trivial restriction on $t_a$, since $[f_c^{-1}(a),1]$ is precisely the subinterval on which the three arguments $x$, $p(x)$ and $f(x)$ already lie in $[a,1]$, so there Eq.~(\ref{eq:completion_genversion}) yields no new value of the generator and becomes a requirement on $t_a$.

Further, notice that Theorem~\ref{thm:completion:1} covers many more cases beyond the $\alpha$-section. Indeed, while the $\alpha$-section is recovered by taking $p(x)=\alpha$, any continuous increasing function $p$ with $p(0) > 0$ fits our framework. To illustrate this, we consider the following example.

\begin{example}\label{example:completion:1}
Let $a \in \left(\frac{3-\sqrt{5}}{2}, 1\right)$, $p(x)=a+(1-a)x$ and 
$f(x)=\max\{0,(2-a)x+a-1\}$ for all $x \in [0,1]$. Then $p$ is continuous, 
strictly increasing with $p(x) \geq a$ for all $x \in [0,1]$, and $f$ is a 
$p$-curve with zero threshold $c=\frac{1-a}{2-a}$ since $f(x)=0$ for all 
$x \in [0,c]$, $f$ is strictly increasing on $[c,1]$, $f(1)=1=p(1)$ and 
$f(x) < \min \{x,p(x)\}$ for all $x \in (c,1)$. Moreover, $p(0)=a>c$ since 
$a>\frac{3-\sqrt{5}}{2}$, so Theorem~\ref{thm:completion:1} applies with 
$a=p(0)$. Since the restriction $f_c(x)=(2-a)x+a-1$ is linear, it has the 
following iterates
$$
f_c^n(x)=(2-a)^n(x-1)+1, \quad \text{for all } x \in [c,1] \text{ and } n \in \mathbb{Z}.
$$
In this case $f^n(a)=\max\{0,1-(2-a)^n(1-a)\}$, which is positive if and 
only if $n<\rho$, where $\rho=\frac{-\ln(1-a)}{\ln(2-a)}$, so 
$n_f = \left\lceil \rho \right\rceil$. Since $a > \frac{3-\sqrt{5}}{2}$ is 
equivalent to $(1-a)(2-a)<1$, we have $\rho>1$ and therefore $2 \leq n_f<+\infty$, 
in accordance with the fact that $c>0$ and, by 
Proposition~\ref{prop:archimedean}, every continuous completion is nilpotent.
Further, the generator is given by
\begin{equation}
t(x) = 
\begin{cases}
\displaystyle t_a\bigl((2-a)^{-n}(x-1)+1\bigr) 
+ \sum_{i=1}^{n} t_a\bigl(1+(1-a)(2-a)^{-i}(x-1)\bigr),
& \text{if } x \in [f^{n}(a), f^{\,n-1}(a)), \\
& \qquad\qquad 1 \le n \le n_f, \\[1em]
t_a(x),
& \text{if } x \in [a, 1],
\end{cases}
\end{equation}
where $t_a:[a,1] \to [0,1]$ is a continuous strictly decreasing function with 
$t_a(1)=0$, $t_a(a)=1$, and satisfying the consistency condition
$$
t_a\bigl((2-a)x+a-1\bigr) = t_a(x) + t_a\bigl(a+(1-a)x\bigr),
\qquad \text{for all } x \in \left[\tfrac{1}{2-a},1\right],
$$
since $f_c^{-1}(a)=\frac{1}{2-a}$.
Taking $t_a(x)=\frac{1-x}{1-a}$, which satisfies all the conditions, the sums 
reduce to
$$\frac{(2-a)^{-n}(1-x)}{1-a} + \sum_{i=1}^n (1-x)(2-a)^{-i} = 
\frac{(2-a)^{-n}(1-x)}{1-a}  + (1-x)\frac{1-(2-a)^{-n}}{1-a} = \frac{1-x}{1-a},$$
where we used the geometric series $\sum_{i=1}^n (2-a)^{-i} = 
\frac{1-(2-a)^{-n}}{1-a}$. This is consistent with the fact that both $f$ and 
$t_a$ were constructed precisely from the \L ukasiewicz t-norm, since 
$f(x)=\TLK(x,p(x))=\max\{x+p(x)-1,0\}$ and $t(x)=\frac{1-x}{1-a}$ is an additive 
generator of \TLK.

Nonetheless, there are other valid choices for $t_a$ that do not lead to the 
\L ukasiewicz t-norm. To show some of them, in this concrete example we can 
exploit the fact that both $p$ and $f_c$ are affine and fix the point $1$, so the 
change of variable $u=1-x$ turns them into multiplications by a constant; that is, 
$1-p(x)=(1-a)(1-x)$ and $1-f_c(x)=(2-a)(1-x)$ for all $x\in[c,1]$. Writing 
$g(u)=t_a(1-u)$ for all $u \in [0,1-a]$, and noting that $f=f_c$ on 
$\left[\frac{1}{2-a},1\right]$, the consistency condition becomes
$$
g\bigl((2-a)u\bigr)=g(u)+g\bigl((1-a)u\bigr),
\qquad \text{for all } u \in \left[0,\tfrac{1-a}{2-a}\right],
$$
which is a linear functional equation in a single variable. The choice 
$t_a(x)=\frac{1-x}{1-a}$ corresponds to the particular solution 
$g(u)=\frac{u}{1-a}$, which satisfies the equation because the two constants $1$ 
and $1-a$ add up to the third one, $2-a$. Looking now for solutions of the form 
$g(u)=\frac{u}{1-a}Q(u)$, the same identity is preserved as soon as 
$Q\colon(0,+\infty)\to(0,+\infty)$ is a continuous function unchanged by both 
multiplications, that is
$$
Q\bigl((2-a)u\bigr)=Q\bigl((1-a)u\bigr)=Q(u), \qquad \text{for all } u>0.
$$
Since multiplication becomes addition in logarithmic scale, we write 
$Q(u)=P(\ln u)$ with $P(s)=Q(e^s)$, so that the two equalities above turn into
$$
P\bigl(s+\ln(2-a)\bigr)=P\bigl(s+\ln(1-a)\bigr)=P(s), \qquad \text{for all } s \in \RR,
$$
that is, $P$ is a continuous function having $\ln(2-a)$ and $-\ln(1-a)$ as 
periods. As the set of periods of a continuous function is a closed subgroup of 
\RR, it is either the whole of \RR, which forces $P$ to be constant, or of the 
form $\omega\ZZ$ for some $\omega>0$. Hence a non-constant $P$ exists if, and 
only if, both periods are integer multiples of a common one, that is, if their 
quotient, which is precisely $\rho$, is rational. Writing $\rho=\frac{l}{k}$ in 
lowest terms and $\omega=\frac{\ln(2-a)}{k}$, we get $\ln(2-a)=k\omega$ and 
$-\ln(1-a)=l\omega$, so any continuous $P$ of period $\omega$ is admissible.
Since $\rho=\frac{-\ln(1-a)}{\ln(2-a)}$ is continuous in $a$ and unbounded, 
rational values occur for infinitely many $a$. For instance, taking $\rho=2$, the 
corresponding $a$ is the unique root in $\left(\frac{3-\sqrt{5}}{2},1\right)$ of 
$(1-a)(2-a)^{2}=1$, approximately $0.5345$, for which $\omega=\ln(2-a)$ and 
$n_f=2$.
Fixing such an $a$ and choosing $P$ to be a sinusoid of period $\omega$, we obtain
\begin{equation}\label{eq:ta:oscillating}
t_a(x)=\frac{1-x}{1-a}\left(1+\varepsilon
\sin\left(\frac{2\pi}{\omega}\ln\frac{1-x}{1-a}\right)\right),
\qquad \text{for all } x \in [a,1),
\end{equation}
with the understanding $t_a(1)=0$, which is the limit at $x=1$ because the sine is 
bounded. This function satisfies $t_a(a)=1$, since $\ln 1=0$, and is strictly 
decreasing whenever 
$|\varepsilon|<\bigl(1+(2\pi/\omega)^{2}\bigr)^{-1/2}$, because in the variable 
$s=\ln\frac{1-x}{1-a}$ it takes the form $e^{s}\bigl(1+\varepsilon\sin(2\pi s/\omega)\bigr)$, 
whose derivative is positive under this bound. Multiplicatively periodic factors 
of this kind are the usual source of non-uniqueness in linear functional 
equations in a single variable \cite{Kuczma1978}. Interestingly, for these 
choices the generator of the corresponding completion is still given by a closed 
formula. Indeed, in the variable $u=1-x$ the arguments appearing in 
Eq.~(\ref{gen:completion:1}) become $(2-a)^{-n}u$ and $(1-a)(2-a)^{-i}u$, at 
which $Q$ takes the value $Q(u)$ by construction, so only the linear factors 
contribute and
$$
(2-a)^{-n}+(1-a)\sum_{i=1}^{n}(2-a)^{-i}
=(2-a)^{-n}+1-(2-a)^{-n}=1.
$$
Consequently, the recursion does not introduce new pieces and $t$ is given on the 
whole interval $[0,1]$ by the same expression as $t_a$ in 
Eq.~(\ref{eq:ta:oscillating}), exactly as in the \L ukasiewicz case 
$\varepsilon=0$. In particular $t(0)=\frac{1}{1-a}$, since $-\ln(1-a)=l\omega$ 
makes the sine vanish, so all these generators share the value at $0$ obtained 
for $\varepsilon=0$. The corresponding t-norms are nilpotent and satisfy 
$T(x,p(x))=f(x)$ for all $x\in[0,1]$, yet for $\varepsilon\neq 0$ they are 
different from \TLK, since otherwise $t$ would be a positive multiple of 
$\frac{1-x}{1-a}$ and $Q$ would be constant. Figure~\ref{fig:example1} displays 
some level curves of these t-norms for different values of $\varepsilon$, where 
we can observe that for $\varepsilon \neq 0$ they coincide with \TLK on the curve 
$(x,p(x))$, whereas elsewhere they differ.

\begin{figure}
\centering
\includegraphics[scale=0.85]{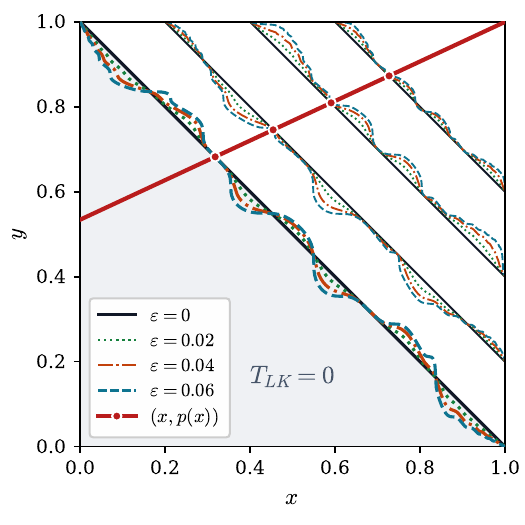}
\caption{Some level curves of the t-norms of Example~\ref{example:completion:1} for $\rho=2$ and different values of $\varepsilon$, the case $\varepsilon=0$ being \TLK. The shaded area is the zero region of \TLK, bounded above by its induced negation, and the remaining lines close to it are the induced negations of the other operators, that is, their level curves of value $0$. The straight red line is the curve $(x,p(x))$, on which all the operators coincide, whereas they differ elsewhere.}\label{fig:example1}
\end{figure}
\end{example}

\section{Continuous completion when $p$ is a suitable path}\label{sec:general}
In this section, we extend the analysis of Section~\ref{sec:constant} to a more general case, in which $p(0) = 0$ is possible. The recursive construction underlying Theorem~\ref{thm:completion:1} relied on the fact that $p(0) > 0$, which guaranteed that $p(x) \geq a$ for all $x \in [0,1]$ for any $a \in (0,p(0)]$. As a consequence, at each step of the recursion the term $p \circ f_c^{-1}(x)$ fell within the region where the generator was already known. When $p(0) = 0$, this property is lost: the values of $p \circ f_c^{-1}(x)$ may be arbitrarily close to zero, and the functional equation can no longer be resolved with our iterative procedure without further assumptions. To preserve the recursive scheme, we therefore need to impose a compatibility condition between $p$ and the iterates of $f$ that ensures $p \circ f_c^{-1}(x)$ always lies in the previously reconstructed region. To make this precise, since $f_c:[c,1] \to [0,p(1)]$ is a continuous strictly increasing function, we may rewrite Eq.~(\ref{eq:completion_genversion}) in the equivalent form
\begin{equation}\label{eq:functional_eq_inverse}
t(x) = t(f_c^{-1}(x)) + t(p \circ f_c^{-1}(x)), \quad 0 \leq x \leq p(1).
\end{equation}
Fix $a \in (0, p(1)] \cap (0,1)$ and suppose that a piece of the generator $t_a : [a,1] \to [0,+\infty)$ is known. For Eq.~(\ref{eq:functional_eq_inverse}) to propagate the values of $t_a$ to the whole interval $[0,a)$ via the partition induced by the iterates of $f$ (see Proposition~\ref{prop:properties-f-with-pcurve}), we need that whenever $x \in [f^n(a), f^{n-1}(a))$ with $1 \leq n \leq n_f$, both $f_c^{-1}(x)$ and $p \circ f_c^{-1}(x)$ lie in the region where $t$ is already determined, namely $[f^{n-1}(a),1]$. The first inclusion is automatic, since $f_c^{-1}(x) \in [f^{n-1}(a), f^{n-2}(a)) \subseteq [f^{n-1}(a),1]$ if $n \geq 2$ or $f_c^{-1}(x) \in [a,1]$ when $n=1$, by (iii)-Proposition~\ref{prop:properties-f-with-pcurve}.

\begin{definition}\label{def:suitable_path}
Let $f$ be a $p$-curve and let $a \in (0, p(1)] \cap (0,1)$. For every $x \in [0,p(1)]$ we define
$$
k_x = \inf \{k \in \{0,\dots,n_f\} \mid p ( f_c^{-1}(x)) \geq f^k(a)\},
$$
that is, $k_x$ is the index of the interval of the partition given in (ii)-Proposition~\ref{prop:properties-f-with-pcurve} containing the value $p(f_c^{-1}(x))$, which is $[f^{k_x}(a),f^{k_x-1}(a))$ when $k_x \geq 1$ and $[a,1]$ when $k_x=0$. We say that $p$ is a \emph{suitable path} if $p(f_c^{-1}(x))$ always lies in the region already covered by the previous steps of the recursion, that is
$$
p ( f_c^{-1}(x)) \geq a, \quad \text{for all } x \in [a,p(1)],
$$
$$
p ( f_c^{-1}(x)) \geq f^{n-1}(a), \quad \text{for all } x \in [f^n(a),f^{n-1}(a)) \text{ and all } n \in \{1,\dots,n_f\}.
$$
Equivalently, $k_x=0$ for all $x \in [a,p(1)]$ and $k_x \leq n-1$ for all $x \in [f^n(a),f^{n-1}(a))$ and all $n \in \{1,\dots,n_f\}$.
\end{definition}
When $p$ is a suitable path, the recursive reconstruction proceeds exactly as in Section~\ref{sec:constant}: on each interval $[f^n(a), f^{n-1}(a))$ both terms on the right-hand side of Eq.~(\ref{eq:functional_eq_inverse}) are already determined from the previous steps, so $t(x)$ is well-defined. However, the generator cannot be written in the same closed form as in Section~\ref{sec:constant}. The reason is that, for $x \in [f^n(a), f^{n-1}(a))$, the value $p \circ f_c^{-1}(x)$ may fall in any interval $[f^{k_x}(a), f^{k_x-1}(a))$ with $k_x \in \{1,\dots, n-1\}$, or in $[a,1]$ when $k_x=0$, so the expression for $t(p \circ f_c^{-1}(x))$ branches according to the value of $k_x$. Since each branch in turn depends on the previously constructed pieces, the generator must be defined recursively: we introduce a family $\{t_a^n\}$ where each $t_a^n$ is defined on $[f^n(a), f^{n-1}(a))$ in terms of the previously constructed pieces $t_a^m$ with $m \leq n-1$, and then assemble them into the candidate generator $t^*$.
\begin{definition}\label{def:operador_generador_sense_restriccions}
Let $f$ be a $p$-curve such that $p$ is a suitable path, let $a \in (0, p(1)] \cap (0,1)$, and let $t_a:[a,1]\to [0,1]$ be a continuous strictly decreasing function with $t_a(1)=0$ and $t_a(a)=1$. We recursively define the family of functions $\{t_a^n\}_{n\in \{0,\dots ,n_f\}}$, with $t_a^0 : [a,1] \to [0,1]$ and $t_a^n : [f^n(a), f^{n-1}(a)) \to [0,+\infty)$ for $n \geq 1$, by
\begin{align*}
    t_a^0(x)&=t_a(x),&x\in[a,1],\\
    t_a^n(x)&=t_a^{n-1}(f_c^{-1}(x))+t_a^{k_x}(p(f_c^{-1}(x))),&x\in [f^n(a),f^{n-1}(a)),\ 1 \leq n \leq n_f.
\end{align*}
The \emph{associated candidate generator} $t^* : [0,1] \to [0, +\infty]$ is defined by assembling these pieces,
\begin{equation}\label{eq:t_star}
    t^*(x) = 
    \begin{cases}
        +\infty, & \text{if } n_f=+\infty \text{ and } x=0, \\
        t_a^n(x), & \text{if } f^n(a) \le x < f^{n-1}(a),\ 1 \le n \le n_f,\\[4pt]
        t_a(x), & \text{if } a \le x \le 1.
    \end{cases}
\end{equation}
\end{definition}
Before characterizing the completions, we have to check that the associated candidate generator is indeed an additive generator. The next result collects the properties of the family $\{t_a^n\}$ and of $t^*$ that are needed for proving this. We prove that each $t_a^n$ is well-defined, so that $t^*$ is well-defined on $[0,1]$, and we give a semi-closed form for it in terms of $t_a$. Moreover, if $t_a$ satisfies a consistency condition at the point $f_c^{-1}(a)$, then $t^*$ is continuous and strictly decreasing with $t^*(1)=0$, that is, it is an additive generator.
\begin{lemma} \label{lema:propietats-previes}
Under the conditions of Definition~\ref{def:operador_generador_sense_restriccions} the following properties hold:
\begin{enumerate}[label=(\roman*)]
    \item The function $t_a^n$ is well-defined for all $n\in\{0,\dots ,n_f\}$ and consequently $t^*$ is well-defined on $[0,1]$.
    \item For each $1 \leq n \leq n_f$ and each $x \in [f^n(a), f^{n-1}(a))$,
    $$t_a^n(x) = t_a(f_c^{-n}(x)) + \sum_{i=1}^n t_a^{k_i}\big(p(f_c^{-i}(x))\big),$$
    where $k_i = k_{f_c^{-(i-1)}(x)}$ for $1 \leq i \leq n$.
\end{enumerate}
If, in addition, $t_a$ satisfies the consistency condition
\begin{equation}\label{eq:consistency_condition}
t_a(p(f_c^{-1}(a))) = 1 - t_a(f_c^{-1}(a)),
\end{equation}
the function $t^*$ given in Eq.~(\ref{eq:t_star}) satisfies the following properties:

\begin{enumerate}[label=(\roman*),start=3]
    \item $t^*$ is continuous on $(0,1]$.
    \item $t^*$ is strictly decreasing on $[0,1]$.
    \item $t^*$ is continuous at $x=0$.
    \item $t^*$ is an additive generator.
\end{enumerate}
\end{lemma}
\begin{proof}
    \begin{enumerate}[label=(\roman*)]
        \item We argue by induction on $n$. The case $n=0$ is immediate, since $t_a^0 = t_a$ is given. For $n=1$, let $x \in [f(a), a)$. Then $f_c^{-1}(x) \in [a, 1]$ lies in the domain of $t_a^0$, so the first term $t_a^0(f_c^{-1}(x))$ is well-defined. For the second term, since $p$ is a suitable path we have $k_x = 0$, and by the definition of $k_x$ the value $p(f_c^{-1}(x))$ lies in $[a, 1]$, which is the domain of $t_a^0$. Therefore $t_a^0(p(f_c^{-1}(x)))$ is well-defined too, and so is $t_a^1(x)$.

Assume now that $n \geq 2$ and that $t_a^m$ is well-defined on its domain for all $0 \leq m \leq n-1$, and let $x \in [f^n(a), f^{n-1}(a))$. Then $f_c^{-1}(x) \in [f^{n-1}(a), f^{n-2}(a))$ lies in the domain of $t_a^{n-1}$, so the first term $t_a^{n-1}(f_c^{-1}(x))$ is well-defined by the induction hypothesis. For the second term, since $p$ is a suitable path we have $k_x \leq n-1$, and by the definition of $k_x$ the value $p(f_c^{-1}(x))$ lies in $[f^{k_x}(a), f^{k_x-1}(a))$, which is the domain of $t_a^{k_x}$ (or in $[a,1]$ if $k_x = 0$). Hence $t_a^{k_x}(p(f_c^{-1}(x)))$ is also well-defined by the induction hypothesis, and therefore $t_a^n(x)$ is well-defined.

Since the intervals $\{[f^n(a), f^{n-1}(a))\}_{n=1}^{n_f} \cup \{[a,1]\}$ form a partition of $(0,1]$, $t^*$ is well-defined on $(0,1]$. Finally, if $n_f<+\infty$ then $0$ belongs to $[f^{n_f}(a),f^{n_f-1}(a))$ and $t^*(0)=t_a^{n_f}(0)$, whereas if $n_f=+\infty$ then $t^*(0)=+\infty$ by definition, so $t^*$ is well-defined on $[0,1]$.
        \item We argue by induction on $n$. For $n = 1$, let $x \in [f(a), a)$. By Definition~\ref{def:operador_generador_sense_restriccions},
$$t_a^1(x) = t_a^0(f_c^{-1}(x)) + t_a^{k_x}(p(f_c^{-1}(x))) = t_a(f_c^{-1}(x)) + t_a^{k_1}(p(f_c^{-1}(x))).$$
Assume now that the identity holds for some $n \geq 1$, and let $x \in [f^{n+1}(a), f^n(a))$. Then $f_c^{-1}(x) \in [f^n(a), f^{n-1}(a))$, so by Definition~\ref{def:operador_generador_sense_restriccions} and the induction hypothesis applied to $f_c^{-1}(x)$,
\begin{align*}
t_a^{n+1}(x) &= t_a^n(f_c^{-1}(x)) + t_a^{k_x}(p(f_c^{-1}(x)))\\
&= t_a(f_c^{-(n+1)}(x)) + \sum_{i=1}^n t_a^{k_{i+1}}(p(f_c^{-(i+1)}(x))) + t_a^{k_1}(p(f_c^{-1}(x)))\\
&= t_a(f_c^{-(n+1)}(x)) + \sum_{i=1}^{n+1} t_a^{k_i}(p(f_c^{-i}(x))).
\end{align*}
        \item We prove by induction on $n \in \{0, \dots, n_f\}$ the following two statements:
\begin{itemize}
    \item[(A$_n$)] The function $t_a^n$ is continuous on its domain.
    \item[(B$_n$)] If $n < n_f$, the following identity holds:
    $$\lim_{x \to f^n(a)^-} t_a^{n+1}(x) = t_a^n(f^n(a)).$$
\end{itemize}

For the base case $n=0$, (A$_0$) is immediate since $t_a^0 = t_a$ is continuous by hypothesis. For (B$_0$), by the definition of $t_a^1$ and the continuity of $t_a$, $p$ and $f_c^{-1}$,
$$\lim_{x \to a^-} t_a^1(x) = \lim_{x \to a^-}\left[t_a(f_c^{-1}(x)) + t_a(p(f_c^{-1}(x)))\right] = t_a(f_c^{-1}(a)) + t_a(p(f_c^{-1}(a))) = 1,$$
where the last equality uses the condition in Eq.~(\ref{eq:consistency_condition}). Since $t_a^0(a) = t_a(a) = 1$, the identity (B$_0$) holds.

Now, assume (A$_m$) and (B$_m$) hold for all $0 \leq m \leq n$, and let us prove (A$_{n+1}$) and (B$_{n+1}$).

\emph{Proof of (A$_{n+1}$).} Fix $x \in [f^{n+1}(a), f^n(a))$. By Definition~\ref{def:operador_generador_sense_restriccions},
$$t_a^{n+1}(x) = t_a^n(f_c^{-1}(x)) + t_a^{k_x}(p(f_c^{-1}(x))).$$
The first summand is continuous in $x$ since $f_c^{-1}$ is continuous and $t_a^n$ is continuous by (A$_n$). For the second summand, note that $k_x$ may take several values within $\{0,\dots,n\}$ as $x$ varies in $[f^{n+1}(a), f^n(a))$. Nonetheless, since $p$, $f_c^{-1}$ and $t_a^{k_x}$ are continuous functions, the only points where the 
continuity of $t_a^{k_x}(p(f_c^{-1}(x)))$ is not immediate are those $x_0\in [f^{n+1}(a),f^n(a))$ where $p(f_c^{-1}(x_0)) = f^{k_{x_0}}(a)$, i.e., points where $k_x$ jumps from $k_{x_0}+1$ to $k_{x_0}$. At any such point, the left and right limits are
$$\lim_{x \to x_0^-} t_a^{k_x}(p(f_c^{-1}(x))) = \lim_{y \to f^{k_{x_0}}(a)^-} t_a^{k_{x_0}+1}(y),$$
$$\lim_{x \to x_0^+} t_a^{k_x}(p(f_c^{-1}(x))) = t_a^{k_{x_0}}(f^{k_{x_0}}(a)),$$
and these agree by (B$_{k_{x_0}}$) since $k_{x_0}\leq n$. Hence $t_a^{n+1}$ is continuous on $[f^{n+1}(a), f^n(a))$.

\emph{Proof of (B$_{n+1}$).} Assume $n+1 < n_f$. We compute the left-hand limit at $x = f^{n+1}(a)$ using Definition~\ref{def:operador_generador_sense_restriccions},
$$\lim_{x \to f^{n+1}(a)^-} t_a^{n+2}(x) = \lim_{x \to f^{n+1}(a)^-}\left[t_a^{n+1}(f_c^{-1}(x)) + t_a^{k_x}(p(f_c^{-1}(x)))\right].$$
Since $f_c^{-1}$ is continuous and strictly increasing, $x \to f^{n+1}(a)^-$ implies $f_c^{-1}(x) \to f^n(a)^-$, hence by (B$_n$),
$$\lim_{x \to f^{n+1}(a)^-} t_a^{n+1}(f_c^{-1}(x)) = t_a^n(f^n(a)).$$
For the second summand, let $z = p(f^n(a))$ and denote $k^* = k_{f^{n+1}(a)}$. Since $p$ is a continuous increasing function, $p(f_c^{-1}(x)) \to z^-$ as $x \to f^{n+1}(a)^-$. We distinguish two cases.
\begin{itemize}
    \item If $z > f^{k^*}(a)$, then $k_x = k^*$ for all $x$ in a left neighborhood of $f^{n+1}(a)$, and by (A$_{k^*}$),
    $$\lim_{x \to f^{n+1}(a)^-} t_a^{k_x}(p(f_c^{-1}(x))) = t_a^{k^*}(z) = t_a^{k^*}(p(f_c^{-1}(f^{n+1}(a)))).$$
    \item If $z = f^{k^*}(a)$, then $k_x = k^* + 1$ for all $x$ in a left neighborhood of $f^{n+1}(a)$, and by (B$_{k^*}$),
    $$\lim_{x \to f^{n+1}(a)^-} t_a^{k_x}(p(f_c^{-1}(x))) = \lim_{y \to f^{k^*}(a)^-} t_a^{k^*+1}(y) = t_a^{k^*}(f^{k^*}(a)) = t_a^{k^*}(p(f_c^{-1}(f^{n+1}(a)))).$$
\end{itemize}
In both cases, applying Definition~\ref{def:operador_generador_sense_restriccions} at $x = f^{n+1}(a)$,
$$\lim_{x \to f^{n+1}(a)^-} t_a^{n+2}(x) = t_a^n(f^n(a)) + t_a^{k^*}(p(f_c^{-1}(f^{n+1}(a)))) = t_a^{n+1}(f^{n+1}(a)),$$
which proves (B$_{n+1}$).

By induction, (A$_n$) and (B$_n$) hold for all $0 \leq n \leq n_f$. Since the intervals $\{[f^n(a), f^{n-1}(a))\}_{n=1}^{n_f} \cup \{[a,1]\}$ form a partition of $(0,1]$, and (B$_n$) ensures continuity at each junction $f^n(a)$, $t^*$ is continuous on $(0,1]$.
        \item We prove by induction on $n \in \{0,\dots,n_f\}$ that $t^*$ is strictly decreasing on $[f^{n}(a),1]$, which gives the result on $(0,1]$ since the intervals $[f^{n}(a),1]$ cover it. The case $n=0$ is immediate, since $t^*$ coincides with $t_a$ on $[a,1]$.

        Assume that $t^*$ is strictly decreasing on $[f^{n-1}(a),1]$ for some $1 \leq n \leq n_f$ and let $x,y \in [f^{n}(a),f^{n-1}(a))$ with $x<y$. By Definition~\ref{def:operador_generador_sense_restriccions} together with Eq.~(\ref{eq:t_star}), and since $k_x$ is by definition the index of the piece containing $p(f_c^{-1}(x))$,
        \begin{equation*}
        t^*(x)=t^*\bigl(f_c^{-1}(x)\bigr)+t^*\bigl(p(f_c^{-1}(x))\bigr).
        \end{equation*}
        Both $f_c^{-1}$ and $p ( f_c^{-1})$ are increasing, so $f_c^{-1}(x)<f_c^{-1}(y)$ and $p(f_c^{-1}(x)) \leq p(f_c^{-1}(y))$. Moreover, both arguments lie in $[f^{n-1}(a),1]$, the first one by (iii)-Proposition~\ref{prop:properties-f-with-pcurve} and the second one because $p$ is a suitable path, so the induction hypothesis applies and gives $t^*(f_c^{-1}(x))>t^*(f_c^{-1}(y))$ together with $t^*(p(f_c^{-1}(x))) \geq t^*(p(f_c^{-1}(y)))$. Adding both relations we obtain $t^*(x)>t^*(y)$, so $t^*$ is strictly decreasing on $[f^{n}(a),f^{n-1}(a))$. As $t^*$ is continuous on $(0,1]$ by (iii) and strictly decreasing on each of the two consecutive intervals $[f^{n}(a),f^{n-1}(a))$ and $[f^{n-1}(a),1]$, it is strictly decreasing on $[f^{n}(a),1]$.

        Finally, if $n_f=+\infty$ then $t^*(0)=+\infty>t^*(x)$ for all $x \in (0,1]$, and if $n_f<+\infty$ then $0$ belongs to $[f^{n_f}(a),f^{n_f-1}(a))$, which is covered by the argument above. Hence $t^*$ is strictly decreasing on $[0,1]$.
        \item We distinguish two cases.
            \begin{itemize}
                \item If $n_f < +\infty$, then $t^*$ is continuous at $x = 0$ since only the right-hand limit exists at that point and we already justified that the function is continuous inside the interval $[0, f^{n_f - 1}(a))$.
                
                \item If $n_f = +\infty$, we show that $\lim_{x \to 0^+} t^*(x) = +\infty$. By (ii), applied to $x = f^n(a)$, using that $f_c^{-n}(f^n(a)) = a$ and $f_c^{-i}(f^n(a)) = f^{n-i}(a)$, and since $f^n(a)$ belongs to the interval of index $n$,
                $$t^*(f^n(a)) = t_a^n(f^n(a)) = t_a(a) + \sum_{i=1}^n t_a^{k_i}(p(f^{n-i}(a))) = 1 + \sum_{i=1}^n t_a^{k_i}(p(f^{n-i}(a))).$$
                By monotonicity of $t^*$ and the fact that $f^{n-i}(a) \leq a$ for all $1 \leq i \leq n$, we have $p(f^{n-i}(a)) \leq p(a)$ and hence
                $$t_a^{k_i}(p(f^{n-i}(a))) = t^*(p(f^{n-i}(a))) \geq t^*(p(a)) > 0,$$
                where $t^*(p(a)) > 0$ since $p(a) < 1$ because $a \in (0,1)$ by the definition of $p$-curve. Hence $t^*(f^n(a)) \geq 1+n\,t^*(p(a))$, which tends to $+\infty$ as $n \to +\infty$. Since $f^n(a) \to 0$ by (ii)-Proposition~\ref{prop:properties-f-with-pcurve} and $t^*$ is decreasing, we conclude that $\lim_{x \to 0^+} t^*(x) = +\infty = t^*(0)$.
            \end{itemize}
        \item It follows from (iii), (iv) and (v) that $t^*$ is a continuous and strictly decreasing function. Furthermore, we know that $t^*(1)=t_a(1)=0$, completing this proof.
    \end{enumerate}
\end{proof}
Under the conditions of Definition~\ref{def:operador_generador_sense_restriccions} and Lemma~\ref{lema:propietats-previes}, given the function $t_a$, its associated candidate generator $t^*$ can be recursively constructed and enjoys the properties collected in Lemma~\ref{lema:propietats-previes}. Note that the recursion only reconstructs the generator below $a$, since on $[f_c^{-1}(a),1]$ the three arguments $x$, $p(x)$ and $f(x)$ already lie in $[a,1]$ and Eq.~(\ref{eq:completion_genversion}) becomes a requirement on $t_a$ itself, just as Eq.~(\ref{eq:consistency_bb}) in Theorem~\ref{thm:completion:1}. These conditions allow us to characterize all continuous t-norms whose $p$-curve section equals $f$.
\begin{theorem}\label{thm:completacio-general}
    Let $f$ be a $p$-curve such that $p$ is a suitable path and ${a \in (0, p(1)] \cap (0,1)}$. Then, $T$ is a continuous t-norm with $T(x,p(x))=f(x)$ for all $x\in [0,1]$ if, and only if, $T$ is an Archimedean t-norm with additive generator given by the associated candidate generator in Eq.~(\ref{eq:t_star}), where $t_a : [a,1] \to [0,1]$ is a continuous, strictly decreasing function with $t_a(1)=0$, $t_a(a)=1$
and
\begin{equation}\label{eq:consistency_general}
t_a(f(x)) = t_a(x) + t_a(p(x)), \quad \text{for all } x \in [f_c^{-1}(a),1].
\end{equation}
\end{theorem}
Note that Eq.~(\ref{eq:consistency_general}) is well posed, since for $x \in [f_c^{-1}(a),1]$ we have $f(x) \in [a,p(1)]$ and, $p$ being a suitable path, $k_{f(x)}=0$, which gives $p(x)=p(f_c^{-1}(f(x))) \geq a$, so the three arguments lie in $[a,1]$. Its particular case $x=f_c^{-1}(a)$ is precisely Eq.~(\ref{eq:consistency_condition}), which is the instance required in Lemma~\ref{lema:propietats-previes}, whereas the remaining ones are needed to recover the values of the operator on the curve for $x \in [f_c^{-1}(a),1]$, exactly as Eq.~(\ref{eq:consistency_bb}) does in Theorem~\ref{thm:completion:1}.

\begin{proof}
    First, let $T$ be a continuous t-norm with $T(x,p(x))=f(x)$ for all $x \in [0,1]$. By Proposition \ref{prop:archimedean} we know that $T$ is Archimedean, so there exists an additive generator $t$ such that $t^{(-1)}(t(x)+t(y)) = T(x,y)$. Without loss of generality, we may assume that $t(a)=1$, otherwise, we replace $t$ with the additive generator $t(x)/t(a)$. We denote by $t_a$ the restriction of $t$ to $[a,1]$, which is continuous, strictly decreasing and satisfies $t_a(1)=0$ and $t_a(a)=1$.

    Let us first check that $t_a$ satisfies Eq.~(\ref{eq:consistency_general}). If $x \in [f_c^{-1}(a),1]$, then $f(x) \geq a$ and, as observed above, $p(x) \geq a$, so the three values $t(x)$, $t(p(x))$ and $t(f(x))$ are given by $t_a$. Since $t$ generates $T$ and $f(x)=T(x,p(x)) \geq a>0$, we obtain $t_a(f(x))=t_a(x)+t_a(p(x))$.
    
    Let $x\in(0,a)$, then it follows from Proposition \ref{prop:properties-f-with-pcurve} that there exists $1\leq n\leq n_f$ such that $x\in [f^n(a),f^{n-1}(a))$, we will prove that $t$ corresponds to the additive generator in Eq.~(\ref{eq:t_star}) by induction over $n$.
    \begin{itemize}
        \item If $n=1$ then $x\in[f(a),a)$ and $f_c^{-1}(x),p ( f_c^{-1}(x))\in [a,1]$, since $k_x=0$ because $p$ is a suitable path, and then
        \begin{align*}
            x=f_c\circ f_c^{-1}(x)&=T(f_c^{-1}(x),p(f_c^{-1}(x)))=t^{(-1)}(t(f_c^{-1}(x))+t(p(f_c^{-1}(x)))),\\
            t(x)&=t_a(f_c^{-1}(x))+t_a(p(f_c^{-1}(x)))=t_a^1(x),
        \end{align*}
        where the second identity follows from the first one because $x>0$.
        \item We assume that Eq.~(\ref{eq:t_star}) holds for $n$ and we prove it for $n+1$. If $x\in [f^{n+1}(a),f^n(a))$ then $f_c^{-1}(x)\in [f^n(a),f^{n-1}(a))$, and
            \begin{align*}
                x=f_c\circ f_c^{-1}(x)&=T(f_c^{-1}(x),p(f_c^{-1}(x)))=t^{(-1)}(t(f_c^{-1}(x))+t(p(f_c^{-1}(x)))),\\
                t(x)&=t(f_c^{-1}(x))+t(p(f_c^{-1}(x)))=t_a^n(f_c^{-1}(x))+t_a^{k_x}(p(f_c^{-1}(x)))=t_a^{n+1}(x).
            \end{align*}
    \end{itemize}
    If $x=0$ then $t(0)=+\infty$ whenever $n_f=+\infty$ by definition, and otherwise 
    \begin{equation*}
        t(0)=t_a^{n_f}(0)<+\infty,
    \end{equation*}
    as it is defined as a finite composition and sum of bounded functions. For the reverse implication we first prove that the associated candidate generator defined in Eq.~(\ref{eq:t_star}) is an additive generator of a continuous, Archimedean t-norm, i.e., we need to prove that $t$ is a continuous strictly decreasing function with $t(1)=0$. This is precisely statement (vi) of Lemma~\ref{lema:propietats-previes}, whose hypotheses are fulfilled, since Eq.~(\ref{eq:consistency_condition}) is the particular case $x=f_c^{-1}(a)$ of Eq.~(\ref{eq:consistency_general}).

Finally, let us prove that if $T$ is a t-norm with the additive generator in Eq.~(\ref{eq:t_star}), then $T(x,p(x))=f(x)$ for all $x\in [0,1]$. We distinguish three cases, according to the position of $x$ in the partition of Proposition~\ref{prop:properties-f-with-pcurve}.

\begin{itemize}
\item If $x \in [0,c]$, which is non-empty only when $c>0$ and hence $n_f<+\infty$, then $f(x)=0$ and $c\geq f^{n_f-1}(a)$. Two more cases appear.
\begin{itemize}
    \item If $n_f=1$, then $c, p(c)\in [a,1]$, since $k_0=0$ because $p$ is a suitable path, and
    \begin{equation*}
    t(x)+t(p(x))\geq t(c)+t(p(c))=t_a(c)+t_a(p(c))=t(0).
    \end{equation*}
    \item If $n_f\geq 2$, then $c<f^{n_f-2}(a)$ and, writing $k_0$ for the value of $k_x$ at $x=0$,
    \begin{align*}
        t(x)+t(p(x))&\geq t(c)+t(p(c)) = t_a^{n_f-1}(c)+t_a^{k_0}(p(c)) =t_a^{n_f-1}(f_c^{-1}(0))+t_a^{k_0}(p(f_c^{-1}(0)))\\
        &=t_a^{n_f}(0)=t(0).
    \end{align*}
\end{itemize}
In either case, we have  $T(x,p(x))=0=f(x)$.
\item If $x \in (c,a)$, then $f_c(x)=f(x)>0$ and there exists a $1 \leq n \leq n_f$ such that $x \in [f^{n}(a),f^{n-1}(a))$. Note that $n \leq n_f-1$ whenever $n_f<+\infty$, since $x>c \geq f^{n_f-1}(a)$. Then $f(x) \in [f^{n+1}(a),f^{n}(a))$ and we have
\begin{eqnarray*}
t(f_c(x)) &=& t_a^{n+1}(f_c(x)) = t_a^n(f_c^{-1}(f_c(x)))+t_a^{k_{f_c(x)}}(p(f_c^{-1}(f_c(x)))) = t_a^n(x)+t_a^{k_{f_c(x)}}(p(x)) \\
& = & t(x) + t(p(x)),
\end{eqnarray*}
where $k_{f_c(x)}$ denotes the value of $k_y$ at $y=f_c(x)$, so that $t_a^{k_{f_c(x)}}(p(x))=t^*(p(x))=t(p(x))$ by Eq.~(\ref{eq:t_star}).
\item If $x \in [a,1]$, we distinguish two subcases. If $x \in [a,f_c^{-1}(a))$, then $f(x) \in [f(a),a)$ and, since $p$ is a suitable path, $k_{f(x)}=0$, so that $p(x)=p(f_c^{-1}(f(x))) \geq a$ and, applying the case $n=1$ of Eq.~(\ref{eq:t_star}),
\begin{equation*}
t(f(x)) = t_a(f_c^{-1}(f(x))) + t_a(p(f_c^{-1}(f(x)))) = t_a(x) + t_a(p(x)) = t(x)+t(p(x)).
\end{equation*}
If $x \in [f_c^{-1}(a),1]$, then $f(x) \geq a$ and, by the same argument, $p(x) \geq a$, so the three values are given by $t_a$ and the identity $t(f(x))=t(x)+t(p(x))$ is exactly Eq.~(\ref{eq:consistency_general}).
\end{itemize}
In the second and third cases $f(x)>0$ and $t(f(x))=t(x)+t(p(x))$, hence $T(x,p(x))=t^{(-1)}(t(x)+t(p(x)))=f(x)$.
\end{proof}

Note that Theorem~\ref{thm:completacio-general} generalizes Theorem~\ref{thm:completion:1}. Indeed, if $p(0)>0$ and $a \in (0,p(0)]$, then $p(x) \geq p(0) \geq a$ for all $x \in [0,1]$, so $k_x=0$ for every $x$ in the domain of $k$ and $p$ is trivially a suitable path. In this case all the pieces $t_a^n$ are built from $t_a^0=t_a$ alone, the recursion of Definition~\ref{def:operador_generador_sense_restriccions} collapses into a closed formula, and the consistency condition is imposed on the same interval as in Section~\ref{sec:constant}, recovering the expression obtained there.
\begin{corollary}
Let $f$ be a $p$-curve with $p(0)>0$ and let $a \in (0,p(0)]$. Then $p$ is a suitable path, Eq.~(\ref{eq:consistency_general}) coincides with Eq.~(\ref{eq:consistency_bb}), and for each $1 \leq n \leq n_f$,
\begin{equation*}
t_a^n(x)=t_a\bigl(f_c^{-n}(x)\bigr)+\sum_{i=1}^{n} t_a\bigl(p(f_c^{-i}(x))\bigr),
\qquad f^{n}(a) \le x < f^{n-1}(a).
\end{equation*}
Consequently, the associated candidate generator of Eq.~(\ref{eq:t_star}) is the additive generator given in Eq.~(\ref{gen:completion:1}) and Theorem~\ref{thm:completacio-general} reduces to Theorem~\ref{thm:completion:1}.
\end{corollary}
\begin{proof}
Since $p$ is increasing, $p(x) \geq p(0) \geq a$ for all $x \in [0,1]$, so $p(f_c^{-1}(x)) \geq a = f^0(a)$ and hence $k_x=0$ for all $x \in [0,p(1)]$, that is, $p$ is a suitable path and we are under the conditions of Theorem~\ref{thm:completacio-general}. In particular, $f_c^{-1}(x)$ and $p(f_c^{-1}(x))$ lie in $[a,1]$ for all $x \in [a,p(1)]$, so Eq.~(\ref{eq:consistency_general}) is imposed on the same interval $[f_c^{-1}(a),1]$ as Eq.~(\ref{eq:consistency_bb}) and both conditions coincide. Finally, all the indices $k_i$ appearing in (ii)-Lemma~\ref{lema:propietats-previes} are zero, so $t_a^{k_i}=t_a$ and
\begin{equation*}
t_a^n(x)=t_a\bigl(f_c^{-n}(x)\bigr)+\sum_{i=1}^{n} t_a\bigl(p(f_c^{-i}(x))\bigr),
\qquad f^{n}(a) \le x < f^{n-1}(a),
\end{equation*}
for each $1 \leq n \leq n_f$, which is precisely Eq.~(\ref{gen:completion:1}).
\end{proof}

Let us highlight that the suitable path condition is the weakest one under which our recursive scheme can be carried out, since it demands exactly that each value $p(f_c^{-1}(x))$ falls in the region already reconstructed, and nothing more. This generality clearly has a price: the index $k_x$ is only known to satisfy $k_x \leq n-1$, so the second summand in Definition~\ref{def:operador_generador_sense_restriccions} may come from any of the previous pieces and $t_a^n$ is given only implicitly, through the recursion itself. If we impose additional conditions on $p$, then the generator can be described more explicitly. The previous corollary is the extreme case $k_x=0$, in which the second summand always comes from $t_a$ and the recursion does not branch at all. We now consider the situation in which the curve lies above the diagonal, that is
\begin{equation}\label{eq:cond-corollary}
    p(x)\geq x, \quad \text{for all } x \in [0,1
\end{equation}
Note that Eq.~(\ref{eq:cond-corollary}) forces $p(1)=1$, and that it implies that $p$ is a suitable path. Indeed, if $x \in [f^n(a),f^{n-1}(a))$ with $1 \leq n \leq n_f$, then $p(f_c^{-1}(x)) \geq f_c^{-1}(x) \geq f^{n-1}(a)$, which is $k_x \leq n-1$, and the same argument for $x \in [a,1]$ gives $k_x=0$. The relevant consequence, however, is that both $f_c^{-1}$ and $p \circ f_c^{-1}$ map each point to a larger one. Each application of the recursion replaces a value of the generator by two values at higher points, so all the terms obtained after $n$ applications lie in $[a,1]$, where $t_a$ is defined. Unlike in the general case, no branch has to be stopped earlier than the others, and the recursion unfolds into a sum of $2^n$ terms.

To make this precise, observe first that Definition~\ref{def:operador_generador_sense_restriccions} gives
\begin{equation}\label{eq:tstar-rec}
t^*(x)=t^*(f_c^{-1}(x))+t^*(p ( f_c^{-1}(x))), \quad \text{for all } x \in [0,1],
\end{equation}
where for $x<a$ this is the recursion itself, whereas for $x \in [a,1]$ it is Eq.~(\ref{eq:consistency_general}) written at the point $f_c^{-1}(x)$. Applying Eq.~(\ref{eq:tstar-rec}) $n$ times, at each step the argument is composed either with $f_c^{-1}$ or with $p \circ f_c^{-1}$, so we obtain one summand for each word of length $n$ in these two functions. Grouping the words by the number $k$ of times that $p \circ f_c^{-1}$ appears, we define the following family of sets
\begin{equation*}
    S_{n,k} = \left\{w=(w_1,\dots,w_n)\ \mid\ w_i\in \{f_c^{-1}, p\circ f_c^{-1}\}, \#\{i:w_i=p\circ f_c^{-1}\}=k\right\},
\end{equation*}
for $n \in \NN$ and $0 \leq k \leq n$, whose cardinality is $\binom{n}{k}$, where we write $g_w=w_1 \circ \cdots \circ w_n$ for the composition associated with the word $w$.
\begin{corollary}\label{cor:diag}
   Let $f$ be a $p$-curve such that $p(x) \geq x$ for all $x \in [0,1]$, and let $a \in (0,1)$. Then $p$ is a suitable path and $T$ is a continuous t-norm with $T(x,p(x))=f(x)$ for all $x\in [0,1]$ if, and only if, $T$ is an Archimedean t-norm with additive generator given by the associated candidate generator in Eq.~(\ref{eq:t_star}) with
    \begin{equation}\label{eq:diag-gen}
    t_a^n(x)=\displaystyle\sum_{k=0}^n \ \sum_{w \in S_{n,k}} t_a (g_w(x)),
    \quad \text{if } f^{n}(a) \le x < f^{\,n-1}(a), \qquad 1 \le n \le n_f,
    \end{equation}
    where $t_a : [a,1] \to [0,1]$ is a continuous, strictly decreasing function with $t_a(1)=0$, $t_a(a)=1$
and satisfying Eq.~(\ref{eq:consistency_general}).   
\end{corollary}
\begin{proof}
    We have already seen above that $p$ is a suitable path, so we are under the conditions of Theorem~\ref{thm:completacio-general}. Applying Eq.~(\ref{eq:tstar-rec}) $n$ times to a point $x \in [f^n(a),f^{n-1}(a))$ we obtain
    \begin{equation*}
        t^*(x)=\sum_{w \in S_{n,0}} t^*(g_w(x))+\sum_{w \in S_{n,1}} t^*(g_w(x))+\cdots+\sum_{w \in S_{n,n}} t^*(g_w(x))=\sum_{k=0}^{n} \sum_{w \in S_{n,k}} t^*(g_w(x)).
    \end{equation*}
    It remains to check that every argument lies in $[a,1]$, so that $t^*$ can be replaced by $t_a$. Each $w_i$ satisfies $w_i(y) \geq f_c^{-1}(y)$ for all $y \in [0,1]$, and both $f_c^{-1}$ and $p \circ f_c^{-1}$ are increasing, so $g_w(x) \geq f_c^{-n}(x) \geq f_c^{-n}(f^n(a)) \geq a$, whereas $g_w(x) \leq 1$ is immediate. This gives Eq.~(\ref{eq:diag-gen}). Note that the two extreme sets are singletons, the composition associated with the only word of $S_{n,0}$ is $f_c^{-n}$ and the one associated with the only word of $S_{n,n}$ is the composition of $p\circ f_c^{-1}$ $n$ times.
\end{proof}
If, in addition to Eq.~(\ref{eq:cond-corollary}), the functions $p$ and $f_c^{-1}$ commute, then the compositions appearing in Eq.~(\ref{eq:diag-gen}) can be simplified. Indeed, given $n\in\mathbb{N}$ and $0\leq k\leq n$, the composition $g_w$ associated with a word $w \in S_{n,k}$ consists of $n$ functions among which $p \circ f_c^{-1}$ appears $k$ times, so moving all the copies of $p$ to the right we obtain $g_w=f_c^{-n}\circ p^k$. Hence the $\binom{n}{k}$ words of $S_{n,k}$ all produce the same composition and the inner sum in Eq.~(\ref{eq:diag-gen}) reduces to a single term with multiplicity $\binom{n}{k}$, which turns that expression into a fully closed one.
\begin{corollary}\label{cor:commute}
   Let $f$ be a $p$-curve such that $p(x) \geq x$ for all $x \in [0,1]$, and let $a \in (0,1)$. If $p$ and $f_c^{-1}$ commute, then $T$ is a continuous t-norm with $T(x,p(x))=f(x)$ for all $x\in [0,1]$ if, and only if, $T$ is an Archimedean t-norm with additive generator given by the associated candidate generator in Eq.~(\ref{eq:t_star}) with
    \begin{equation}\label{eq:cor:diag-eq-simplificada}
    t_a^n(x)= 
         \displaystyle\sum_{k=0}^n \ \binom{n}{k}t_a(f_c^{-n}( p^k(x))),
    \quad \text{if } f^{n}(a) \le x < f^{\,n-1}(a), \qquad 1 \le n \le n_f,
    \end{equation} 
    where $t_a : [a,1] \to [0,1]$ is a continuous, strictly decreasing function with $t_a(1)=0$, $t_a(a)=1$
and satisfying Eq.~(\ref{eq:consistency_general}).   
\end{corollary}
\begin{proof}
    We can directly apply Corollary \ref{cor:diag} together with the identity $g_w=f_c^{-n}\circ p^k$, valid for every $w \in S_{n,k}$, obtaining
    \begin{equation*}
        t_a^n(x)=\displaystyle\sum_{k=0}^n \ \sum_{w \in S_{n,k}} t_a (g_w(x))=\displaystyle\sum_{k=0}^n \ \binom{n}{k} t_a( f_c^{-n}( p^k(x))),\quad \text{if } f^{n}(a) \le x < f^{\,n-1}(a), \qquad 1 \le n \le n_f.
    \end{equation*}
\end{proof}

As a last corollary, we recover the diagonal, which is the second classical case mentioned in Section~\ref{sec:parametric}. Recall that Theorem~\ref{thm:completion:1} already recovered the $\alpha$-section of Theorem~\ref{theorem:alsina:horizontal} through Corollary~\ref{cor:horizontal}, whereas the diagonal was left aside, since taking $p(x)=x$ gives $p(0)=0$ and the recursion of Section~\ref{sec:constant} does not apply. This case is now covered, because $p(x)=x$ satisfies $p(x) \geq x$ and commutes with $f_c^{-1}$, so Corollary~\ref{cor:commute} applies and Theorem~\ref{theorem:alsina:diagonal} is obtained as a particular case of Theorem~\ref{thm:completacio-general}.
\begin{corollary}\label{cor:diagonal}
    Let $f$ be a $p$-curve such that $p(x)=x$ for all $x\in [0,1]$ and let $a \in (0,1)$. Then $T$ is a continuous t-norm with $T(x,p(x))=f(x)$ for all $x\in [0,1]$ if, and only if, $T$ is an Archimedean t-norm with additive generator given by the associated candidate generator in Eq.~(\ref{eq:t_star}) with
    \begin{equation}\label{eq:cor:p(x)=x}
    t_a^n(x)= 
          2^n t_a(f_c^{-n}(x)),
    \quad \text{if } f^{n}(a) \le x < f^{\,n-1}(a), \qquad 1 \le n \le n_f,
    \end{equation}
    where $t_a : [a,1] \to [0,1]$ is a continuous, strictly decreasing function with $t_a(1)=0$, $t_a(a)=1$ and satisfying Eq.~(\ref{eq:consistency_general}), which in this case is reduced to $t_a(f(x))=2\,t_a(x)$ for all $x \in [f_c^{-1}(a),1]$.
\end{corollary}
\begin{proof}
    In this case, we can apply Corollary \ref{cor:commute}, so we obtain 
    \begin{equation*}
        t_a^n(x)= 
         \displaystyle\sum_{k=0}^n \ \binom{n}{k}t_a(f_c^{-n}( p^k(x)))=
         t_a(f_c^{-n}(x))\displaystyle\sum_{k=0}^n \ \binom{n}{k}=2^nt_a(f_c^{-n}(x)),
    \end{equation*}
    if $ f^{n}(a) \le x < f^{\,n-1}(a),\ 1 \le n \le n_f$, since $p^k=\mathrm{id}$ for all $k$.
\end{proof}
Eq.~(\ref{eq:cor:p(x)=x}) has the same form as Eq.~(\ref{eq:schroder}), the only apparent difference being the interval on which the arbitrary piece of the generator is chosen, namely $[x_1,x_0]$ in Theorem~\ref{theorem:alsina:diagonal} and $[a,1]$ here. The two descriptions agree because $t_a$ is not free on all of $[a,1]$. Indeed, for $p(x)=x$ the consistency condition Eq.~(\ref{eq:consistency_general}) becomes $t_a(f(x))=2\,t_a(x)$ for all $x \in [f_c^{-1}(a),1]$, which can be solved for $t_a(x)$ and, proceeding as in the proof of Theorem~\ref{theorem:alsina:diagonal}, iterated along the increasing sequence $f_c^{-n}(a)$. In this way $t_a$ is determined on the whole of $[a,1)$ by its restriction to $[a,f_c^{-1}(a)]$, which is arbitrary apart from being continuous and strictly decreasing with $t_a(a)=1$ and $t_a(f_c^{-1}(a))=\frac{1}{2}$, and $t_a(1)=0$. Hence the arbitrary part is a single interval in both descriptions, and taking $x_1=a$, so that $x_0=f_c^{-1}(a)$, and $\overline{t}=2t_a$, we recover Eq.~(\ref{eq:schroder}).

Let us highlight that this is possible because the two arguments $x$ and $p(x)$ in Eq.~(\ref{eq:consistency_general}) coincide when $p$ is the identity, so the condition involves a single unknown value and can be solved upward. For any other $p$ the three arguments are distinct, the value obtained depends on the position of $p$ with respect to the diagonal, and the same argument does not apply. Whether the freedom left by Theorem~\ref{thm:completacio-general} admits a description on a single interval for curves other than the diagonal is therefore an open question, which we leave as future work.

We close this section with a concrete example to illustrate the results obtained, in which $p$ and $f$ are taken to be power functions.

\begin{example}\label{example:aczel-alsina}
Let $0<\alpha\le 1$, $r>0$, $p(x)=x^{\alpha}$ and $f(x)=x^{\beta}$ for all
$x\in[0,1]$, where $\beta=(1+\alpha^{r})^{1/r}$. Note that $\beta>1$ since
$\alpha>0$, and $\beta>\alpha$ since $\beta^{r}=1+\alpha^{r}>\alpha^{r}$. Then $p$
is continuous, strictly increasing with $p(0)=0$, $p(1)=1$ and $p(x)\in(0,1)$
for all $x\in(0,1)$, and $f$ is a $p$-curve with zero threshold $c=0$. Indeed,
$f(0)=0$, $f(1)=1=p(1)$, $f$ is strictly increasing, and for all $x\in(0,1)$
both $x^{\beta}<x$ (since $\beta>1$) and $x^{\beta}<x^{\alpha}$ (since
$\beta>\alpha$), so $f(x)<\min\{x,p(x)\}$. As $c=0$ we have $n_f=+\infty$ and
every continuous completion is strict.

The restriction $f_c=f$ is invertible with $f_c^{-1}(x)=x^{1/\beta}$ and
$p( f_c^{-1}(x))=x^{\alpha/\beta}$, and the iterates are
$$
f^{n}(x)=x^{\beta^{n}},\qquad
f_c^{-n}(x)=x^{1/\beta^{n}},\qquad
p^{n}(x)=x^{\alpha^{n}},
$$
for all $x\in[0,1]$ and $n\in\mathbb{N}$. Fixing $a\in(0,1)$, the iterates of $a$ under $f$ are
$f^{n}(a)=a^{\beta^{n}}$, so the partition of $(0,1]$ from
Proposition~\ref{prop:properties-f-with-pcurve} is
$$
(0,1]=\bigcup_{n=1}^{\infty}\bigl[a^{\beta^{n}},\,a^{\beta^{n-1}}\bigr)\;\cup\;[a,1].
$$
Since $\alpha \leq 1$ we have $p(x)=x^{\alpha}\geq x$ for all $x \in [0,1]$, so $p$ is a suitable path. We compute $k_x$ explicitly, which shows that the bound $k_x \leq n-1$ is in general not attained. Since $\log_a$ reverses inequalities as $a<1$,
$$
k_x=\inf\Bigl\{k\ge 0:\ x^{\alpha/\beta}\ge a^{\beta^{k}}\Bigr\}
   =\inf\Bigl\{k\ge 0:\ \beta^{k}\ge \tfrac{\alpha}{\beta}\,\log_a x\Bigr\}.
$$
Applying $\log_{\beta}$, we obtain
$\beta^{k}\ge\tfrac{\alpha}{\beta}\log_a x$ if, and only if,
$k\ge\log_{\beta}\!\bigl(\tfrac{\alpha}{\beta}\log_a x\bigr)$, so
$$
k_x=\max\Bigl\{0,\ \bigl\lceil
\log_{\beta}\!\bigl(\tfrac{\alpha}{\beta}\,\log_a x\bigr)\bigr\rceil\Bigr\}.
$$
For $x\in[f^{n}(a),f^{n-1}(a))$ we have
$\log_a x\in\bigl(\beta^{n-1},\beta^{n}\bigr]$, hence
$$
\log_{\beta}\!\bigl(\tfrac{\alpha}{\beta}\,\log_a x\bigr)
\le\log_{\beta}\!\bigl(\tfrac{\alpha}{\beta}\beta^{n}\bigr)
=(n-1)+\log_{\beta}\alpha\le n-1,
$$
since $\log_{\beta}\alpha\le 0$ as $\alpha\le 1<\beta$. For $\alpha=1$ we have
$\log_{\beta}\alpha=0$ and $k_x=n-1$ exactly, recovering the diagonal $p(x)=x$.

Moreover, $p$ and $f_c^{-1}$ commute,
$$
p( f_c^{-1}(x))=\bigl(x^{1/\beta}\bigr)^{\alpha}
=\bigl(x^{\alpha}\bigr)^{1/\beta}=f_c^{-1}( p(x)),
\qquad \text{for all } x\in[0,1],
$$
so Corollary~\ref{cor:commute} applies and the generator is given by
\begin{equation}
t(x)=
\begin{cases}
+\infty, & \text{if } x=0,\\[1em]
\displaystyle\sum_{k=0}^{n}\binom{n}{k}\,t_a\!\bigl(x^{\alpha^{k}/\beta^{n}}\bigr),
& \text{if } x\in[f^{n}(a),f^{n-1}(a)),\ n\ge 1,\\[1em]
t_a(x), & \text{if } x\in[a,1],
\end{cases}
\end{equation}
where $t_a:[a,1]\to[0,1]$ is a continuous strictly decreasing function with
$t_a(1)=0$, $t_a(a)=1$, and satisfying the consistency condition
\begin{equation}\label{eq:consistency_AA}
t_a\!\bigl(x^{\beta}\bigr)=t_a(x)+t_a\!\bigl(x^{\alpha}\bigr),
\qquad \text{for all } x\in\bigl[a^{1/\beta},1\bigr],
\end{equation}
since $f_c^{-1}(a)=a^{1/\beta}$.

Taking $t_a(x)=\bigl(\log_a x\bigr)^{r}$, which is continuous, strictly
decreasing on $[a,1]$ with $t_a(1)=0$, $t_a(a)=1$, and satisfies
Eq.~(\ref{eq:consistency_AA}) since, by the homogeneity
$t_a(x^{c})=c^{r}t_a(x)$ for all $c>0$,
$$
t_a(x)+t_a\!\bigl(x^{\alpha}\bigr)=(1+\alpha^{r})\,t_a(x)=\beta^{r}t_a(x)
=t_a\!\bigl(x^{\beta}\bigr),
$$
the sum reduces by the same homogeneity together with the
binomial theorem,
$$
\sum_{k=0}^{n}\binom{n}{k}\,t_a\!\bigl(x^{\alpha^{k}/\beta^{n}}\bigr)
=\frac{t_a(x)}{\beta^{nr}}\sum_{k=0}^{n}\binom{n}{k}\alpha^{kr}
=\frac{t_a(x)}{\beta^{nr}}(1+\alpha^{r})^{n}
=t_a(x),
$$
where the last equality uses $\beta^{r}=1+\alpha^{r}$. Hence $t(x)=(\log_a x)^{r}$
for all $x\in(0,1]$ and the reconstructed t-norm is the Acz\'el--Alsina
t-norm of parameter $r$ described in Example~\ref{ex:tnorms},
$$
T(x,y)=t^{-1}\bigl(t(x)+t(y)\bigr)
=a^{\bigl((\log_a x)^{r}+(\log_a y)^{r}\bigr)^{1/r}},
\qquad \text{for all } x,y\in[0,1],
$$
which is consistent with the fact that both $f$ and $t_a$ were constructed
precisely from it, since $T^{AA}_r(x,x^{\alpha})=x^{(1+\alpha^{r})^{1/r}}=x^{\beta}$.
For $r=1$ we obtain $\beta=1+\alpha$ and we recover the product t-norm \TP.

Nonetheless, there are other valid choices of $t_a$. Assume $\alpha<1$ and
observe that, in the variable $u=\log_a x$, the maps $x \mapsto x^{\alpha}$ and
$x \mapsto x^{\beta}$ are given by $u \mapsto \alpha u$ and $u \mapsto \beta u$,
since $\log_a(x^{c})=c\log_a x$ for all $c>0$. Writing $g(u)=t_a(a^{u})$ for all
$u \in [0,1]$, Eq.~(\ref{eq:consistency_AA}) becomes
$$
g(\beta u)=g(u)+g(\alpha u), \qquad \text{for all } u\in\bigl[0,\tfrac{1}{\beta}\bigr],
$$
which is a linear functional equation in a single variable. The choice
$t_a(x)=(\log_a x)^{r}$ corresponds to the particular solution $g(u)=u^{r}$, which
satisfies the equation because $\beta^{r}=1+\alpha^{r}$. Looking now for solutions
of the form $g(u)=u^{r}Q(u)$, the same identity is preserved when
$Q\colon(0,+\infty)\to(0,+\infty)$ is a continuous function which remains unchanged by both
multiplications, that is
$$
Q(\beta u)=Q(\alpha u)=Q(u), \qquad \text{for all } u>0.
$$
As in Example~\ref{example:completion:1}, writing $Q(u)=P(\ln u)$ turns these
equalities into the requirement that $P$ be a continuous function having
$\ln\beta$ and $-\ln\alpha$ as periods, so a non-constant $P$ exists if, and
only if, their quotient $\rho=\frac{-\ln\alpha}{\ln\beta}$ is rational. Writing
$\rho=\frac{l}{m}$ in lowest terms and $\omega=\frac{\ln\beta}{m}$, we get
$\ln\beta=m\omega$ and $-\ln\alpha=l\omega$, so any continuous $P$ of period
$\omega$ is admissible. Since $\rho=\frac{-r\ln\alpha}{\ln(1+\alpha^{r})}$
is continuous in $\alpha$ and unbounded, rational values occur for infinitely many
$\alpha$. For instance, $\rho=1$ amounts to $\alpha\beta=1$, that is
$\alpha^{2r}+\alpha^{r}=1$, whose solution is
$\alpha=\bigl(\frac{\sqrt{5}-1}{2}\bigr)^{1/r}$, with $\omega=\ln\beta$.
Fixing such an $\alpha$ and choosing $P$ to be a sinusoid of period $\omega$, we
obtain
\begin{equation}\label{eq:ta:oscillating:AA}
t_a(x)=\bigl(\log_a x\bigr)^{r}
\left(1+\varepsilon\sin\left(\frac{2\pi}{\omega}\ln\log_a x\right)\right),
\qquad \text{for all } x \in [a,1),
\end{equation}
with the understanding $t_a(1)=0$, which is the limit at $x=1$ because the sine is bounded.
This function satisfies $t_a(a)=1$, since $\ln 1=0$, and is strictly decreasing
whenever $|\varepsilon|<\bigl(1+(2\pi/(r\omega))^{2}\bigr)^{-1/2}$, because in the
variable $s=\ln\log_a x$ it takes the form
$e^{rs}\bigl(1+\varepsilon\sin(2\pi s/\omega)\bigr)$, whose derivative is positive
under this bound. As in Example~\ref{example:completion:1}, for these choices the
generator of the corresponding completion is still given by a closed formula.
Indeed, in the variable $u=\log_a x$ the arguments appearing in the binomial sum
satisfy $\ln\bigl(\alpha^{k}u/\beta^{n}\bigr)=\ln u-(kl+nm)\omega$, so $Q$ takes the value
$Q(u)$ at all of them, only the factors $\bigl(\alpha^{k}u/\beta^{n}\bigr)^{r}$
contribute and the sum collapses as before. Consequently, $t$ is given on the whole
of $(0,1]$ by the same expression as $t_a$ in Eq.~(\ref{eq:ta:oscillating:AA}),
with $t(0)=+\infty$. The corresponding t-norms are strict and satisfy
$T(x,x^{\alpha})=x^{\beta}$ for all $x\in[0,1]$, yet for $\varepsilon\neq 0$ they
differ from $T^{AA}_{r}$, since otherwise $t$ would be a positive multiple of
$(\log_a x)^{r}$ and $Q$ would be constant. Figure~\ref{fig:example2} displays some level curves of these t-norms for different values of $\varepsilon$, where it can be seen that for $\varepsilon \neq 0$ they coincide with $T^{AA}_{r}$ on the curve $(x,p(x))$, whereas elsewhere they differ.

\begin{figure}
\centering
\includegraphics[scale=0.85]{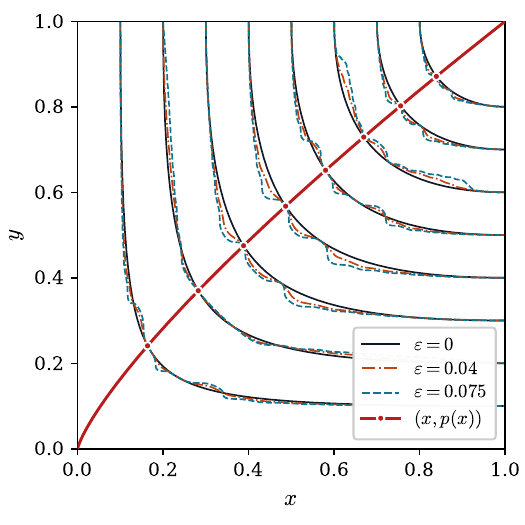}
\caption{Some level curves of the t-norms of Example~\ref{example:aczel-alsina} for $r=2$, $\rho=1$ and different values of $\varepsilon$, the case $\varepsilon=0$ being the Acz\'el--Alsina t-norm $T^{AA}_{2}$. All these operators are strict, so they only vanish on the boundary of the unit square. The red curve is $(x,p(x))$ with $p(x)=x^{\alpha}$, on which all of them coincide, whereas they differ elsewhere, and the marked points are its intersections with the level curves displayed.}\label{fig:example2}
\end{figure}
\end{example}

\section{Unique completion from a $p$-curve and a level curve}\label{sec:unique}

We have already discussed that the completions obtained in Sections~\ref{sec:constant} and~\ref{sec:general} are not unique, since they depend on the choice of $t_a$ on $[a,1]$. However, this is precisely the piece of the generator that a pre-t-norm known above a level curve of value $a$ determines, by Theorem~\ref{thm:generatorabovelevelcurve}, and which in turn leaves the operator undetermined below that curve. The two regions are therefore complementary, and we devote this section to combining them. Our first result shows that they admit a common completion exactly when $t_a$ satisfies Eq.~(\ref{eq:consistency_general}), in which case the completion is unique.

\begin{proposition}\label{prop:compatibility}
Let $a \in (0,1)$, let $l\colon[a,1]\to[a,1]$ be a continuous strictly decreasing function with $l^{-1}=l$, $l(a)=1$ and $l(1)=a$, and let $\TA\colon A \to [0,1]$ be a continuous cancellative pre-t-norm defined on $A=\{(x,y)\in[a,1]^2 \mid y \geq l(x)\}$ with $\TA(x,l(x))=a$ for all $x \in [a,1]$. Let $t_a$ be the additive generator of $\TA$ with $t_a(a)=1$ given by Theorem~\ref{thm:generatorabovelevelcurve}, and let $f$ be a $p$-curve with $a \in (0,p(1)]$ such that $p$ is a suitable path. Then the following statements are equivalent:
\begin{enumerate}[label=(\roman*)]
\item There exists a continuous t-norm $T$ such that $T\mid_A=\TA$ and $T(x,p(x))=f(x)$ for all $x \in [0,1]$.
\item The function $t_a$ satisfies Eq.~(\ref{eq:consistency_general}).
\end{enumerate}
In this case $T$ is Archimedean, its additive generator is the associated candidate generator of Eq.~(\ref{eq:t_star}) built from $t_a$, and it is the unique continuous t-norm satisfying (i).
\end{proposition}

\begin{proof}
Assume (i) and let $x \in [f_c^{-1}(a),1]$. Then $x \geq a$ and, since $p$ is a suitable path, $k_{f(x)}=0$, so $p(x)=p(f_c^{-1}(f(x))) \geq a$. Moreover $T(x,p(x))=f(x) \geq f(f_c^{-1}(a))=a$, hence $p(x) \geq l(x)$ and $(x,p(x)) \in A$. Since $t_a$ generates $\TA$ on $A$ and $t_a(x)+t_a(p(x)) \leq 1$, we obtain
$$t_a(f(x))=t_a\bigl(\TA(x,p(x))\bigr)=t_a(x)+t_a(p(x)),$$
which is (ii).

Assume now (ii). By Theorem~\ref{thm:completacio-general} there exists a continuous Archimedean t-norm $T$ whose additive generator $t$ is the associated candidate generator built from $t_a$ and which satisfies $T(x,p(x))=f(x)$ for all $x \in [0,1]$. If $(x,y) \in A$ then $t_a(x)+t_a(y) \leq 1 < t(0)$, so
$$T(x,y)=t^{(-1)}\bigl(t_a(x)+t_a(y)\bigr)=t_a^{-1}\bigl(t_a(x)+t_a(y)\bigr)=\TA(x,y),$$
that is, $T\mid_A=\TA$ and (i) holds.

Finally, let $T$ be any continuous t-norm satisfying (i). By Proposition~\ref{prop:archimedean} it is Archimedean, so by Theorem~\ref{th:representation_archimedeantnorms} it admits a continuous additive generator $t$, unique up to a positive multiplicative constant, which we normalize so that $t(a)=1$. Since $T(x,p(x))=f(x)$ for all $x \in [0,1]$, Theorem~\ref{thm:completacio-general} gives that $t$ is the associated candidate generator built from its own restriction to $[a,1]$. Since $T\mid_A=\TA$, this restriction is an additive generator of $\TA$ with value $1$ at $a$, so it coincides with $t_a$ by Theorem~\ref{thm:generatorabovelevelcurve}. Hence $t$, and therefore $T$, is uniquely determined.
\end{proof}

The above proposition uses the region above a level curve as a way of determining the function $t_a$ left arbitrary in Theorem~\ref{thm:completacio-general}, so one may wonder whether every such function can be obtained in this way. The following result shows that this is indeed the case, since from any continuous, strictly decreasing $t_a\colon[a,1]\to[0,1]$ with $t_a(a)=1$ and $t_a(1)=0$ one can construct a level curve and a pre-t-norm above it whose additive generator is precisely $t_a$.

\begin{proposition}\label{prop:converse}
Let $a \in (0,1)$ and let $t_a\colon[a,1]\to[0,1]$ be a continuous, strictly decreasing function with $t_a(a)=1$ and $t_a(1)=0$. Define
\begin{equation*}
l(x)=t_a^{-1}\bigl(1-t_a(x)\bigr), \qquad \text{for all } x \in [a,1],
\end{equation*}
and let $A=\{(x,y)\in[a,1]^2 \mid y \geq l(x)\}$. Then $l$ is a continuous, strictly decreasing function with $l^{-1}=l$, $l(a)=1$ and $l(1)=a$, and
\begin{equation*}
\TA(x,y)=t_a^{-1}\bigl(t_a(x)+t_a(y)\bigr), \qquad \text{for all } (x,y) \in A,
\end{equation*}
is a continuous, cancellative pre-t-norm with $\TA(x,l(x))=a$ for all $x \in [a,1]$ whose additive generator with $t_a(a)=1$ is $t_a$.
\end{proposition}
\begin{proof}
Since $t_a$ is a continuous, strictly decreasing bijection from $[a,1]$ onto $[0,1]$, the value $1-t_a(x)$ belongs to $[0,1]$ for all $x \in [a,1]$ and $l$ is well defined. It is continuous and strictly decreasing, because it is the composition of $t_a$, the function $u \mapsto 1-u$ and $t_a^{-1}$, and $l(a)=t_a^{-1}(0)=1$, $l(1)=t_a^{-1}(1)=a$. Moreover $t_a(l(l(x)))=1-t_a(l(x))=t_a(x)$, so $l(l(x))=x$ and $l^{-1}=l$.

Since $t_a$ is decreasing, for $(x,y) \in [a,1]^2$ the inequality $y \geq l(x)$ holds if, and only if, $t_a(y) \leq t_a(l(x))=1-t_a(x)$, that is, if and only if $t_a(x)+t_a(y) \leq 1$. Hence the sum lies in the image of $t_a$ for every $(x,y) \in A$ and $\TA$ is well defined. Commutativity, associativity and monotonicity follow from the representation, the first two because $t_a(x)+t_a(y)+t_a(z) \leq 1$ is symmetric in the three variables, and $\TA(x,1)=t_a^{-1}(t_a(x))=x$ since $t_a(1)=0$. It is cancellative because $t_a$ is injective, continuous because $t_a$ and $t_a^{-1}$ are continuous, and $\TA(x,l(x))=t_a^{-1}(1)=a$ for all $x \in [a,1]$ by the definition of $l$. Finally, $t_a$ is an additive generator of $\TA$ with $t_a(a)=1$, and it is the only one by Theorem~\ref{thm:generatorabovelevelcurve}.
\end{proof}

Note that the previous construction starts only from $t_a$, so the pre-t-norm obtained assigns to the points $(x,p(x))$ of the curve lying in $A$ the values $t_a^{-1}(t_a(x)+t_a(p(x)))$, which coincide with $f(x)$ exactly when $t_a$ satisfies Eq.~(\ref{eq:consistency_general}). Combining both propositions we obtain the following characterization.

\begin{corollary}
Under the hypotheses of Proposition~\ref{prop:converse}, let $f$ be a $p$-curve with $a \in (0,p(1)]$ such that $p$ is a suitable path. Then there exists a continuous t-norm $T$ with $T\mid_A=\TA$ and $T(x,p(x))=f(x)$ for all $x \in [0,1]$ if, and only if, $t_a$ satisfies Eq.~(\ref{eq:consistency_general}), in which case $T$ is unique.
\end{corollary}

Therefore, choosing the function $t_a$ is the same as fixing the operator above a level curve of value $a$, and the two descriptions of the family of completions are equivalent, each choice of one of them determining the other.

We finish this section by revisiting the two examples of Sections~\ref{sec:constant} and~\ref{sec:general}, where the values on the curve left a whole family of completions. Adding the values above a level curve of value $a$ selects exactly one member of each family.

\begin{example}\label{example:unique:lk}
Let $a$, $p$ and $f$ be as in Example~\ref{example:completion:1} and let $l(x)=1+a-x$ for all $x \in [a,1]$, which is the level curve of value $a$ of \TLK. Taking $\TA=\TLK\mid_A$, Theorem~\ref{thm:generatorabovelevelcurve} gives $t_a(x)=\frac{1-x}{1-a}$ for all $x \in [a,1]$, which satisfies Eq.~(\ref{eq:consistency_general}) as we checked there, so Proposition~\ref{prop:compatibility} gives \TLK{} as the unique continuous completion. In particular, the operators of Figure~\ref{fig:example1} with $\varepsilon \neq 0$ are discarded.
\end{example}

\begin{example}\label{example:unique:aa}
Let $a$, $r$, $p$ and $f$ be as in Example~\ref{example:aczel-alsina} and let $l(x)=a^{\left(1-(\log_a x)^{r}\right)^{1/r}}$ for all $x \in [a,1]$, which is the level curve of value $a$ of $T^{AA}_{r}$. Taking $\TA=T^{AA}_{r}\mid_A$, Theorem~\ref{thm:generatorabovelevelcurve} gives $t_a(x)=(\log_a x)^{r}$ for all $x \in [a,1]$, which satisfies Eq.~(\ref{eq:consistency_general}) as we checked there, so Proposition~\ref{prop:compatibility} gives $T^{AA}_{r}$ as the unique continuous completion, and the operators of Figure~\ref{fig:example2} with $\varepsilon \neq 0$ are discarded.
\end{example}

\section{Conclusions}\label{sec:conclusions}
In this paper we have studied the continuous completion problem for Archimedean t-norms known along a parametric curve. We have first formulated the problem for an arbitrary curve $\gamma$ in terms of the functional equation satisfied by an additive generator, and we have then focused on the case $\gamma(\theta)=(\theta,p(\theta))$ with $p$ continuous and increasing, a setting that generalizes the classical diagonal and horizontal section cases. Within it, we have considered first the case where $p$ is bounded away from zero, which is easier to solve, provides a closed expression for the additive generator and already generalizes the horizontal section case. We have then extended the analysis to a more general setting by imposing the suitable path condition, which is what allows us to still use a recursive argument. In both cases we have characterized all the continuous t-norms with a given $p$-curve section, the second result covering the first one as well as the two classical cases.

Our results show that the recursive argument on which they rely is a simple idea, but it turns out to be considerably more flexible and powerful than the classical results may suggest, since it can be applied to a large family of curves and not only to the diagonal and the horizontal section. Further, although in the more general case the expression of the additive generator is given in terms of a recursion, we have discussed several particular cases in which a closed formula can be achieved, namely when the curve lies above the diagonal, and more explicitly when in addition $p$ and the inverse of the corresponding section commute.

As in the classical cases, the completions are not unique, and we have provided concrete examples in which infinitely many continuous t-norms share the same section along a curve, which reveals an interesting behavior of these operators. Furthermore, we have shown that our results combine naturally with the reconstruction of a pre-t-norm known above a level curve, and we have proved that the piece of the generator left arbitrary by our theorems is exactly the information that such a region provides. Consequently, both pieces of information together determine a unique completion.

There are, however, some limitations of our approach that suggest natural directions for future work. The main one is the functional equation that the arbitrary piece of the generator has to satisfy, which must be solved separately in each case, as we have done in the examples. We have discussed this in comparison with the diagonal, where the equation can be solved and the arbitrary part reduces to a single subinterval, so a possible direction for future work is to look for more general results without this limitation, which will require studying which conditions must be imposed on $p$, and in particular its position with respect to the diagonal. A second limitation is the class of curves considered, since we have restricted ourselves to those of the form $(x,p(x))$ with $p$ continuous and increasing, so another direction is to relax these conditions or to consider more general curves.

Finally, many regions of the unit square have already been studied in the completion problem, and our results show that combining two of them can determine the operator even when neither does so on its own. This suggests looking more systematically at how the information about a continuous Archimedean t-norm is distributed over the unit square, identifying which regions determine which part of the additive generator, and how the corresponding results can be combined.

\section*{Declaration of competing interest}

The authors declare that they have no known competing financial
interests or personal relationships that could have appeared to
influence the work reported in this paper.

\section*{Acknowledgments}

This work was partially supported by the R+D+i Project PID2020-113870GB-I00-“Desarrollo de herramientas de Soft Computing para la Ayuda al Diagnóstico Clínico y a la Gestión de Emergencias (HESOCODICE)”, funded by MICIU /AEI/10.13039/501100011033/. Raquel Fernandez-Peralta is funded by the EU NextGenerationEU through the Recovery and Resilience Plan for Slovakia under the project No. 09I03-03-V04- 00557.

\bibliographystyle{cas-model2-names}
\bibliography{cas-refs.bib}

\end{document}